\documentclass[runningheads]{llncs}
\usepackage[T1]{fontenc}
\usepackage{graphicx} 
\usepackage{amsmath}
\usepackage{amssymb}
\usepackage{bussproofs}
\usepackage{enumitem}
\usepackage{comment}
\usepackage{stmaryrd}
\usepackage{framed}

\newcommand{\f}{\bullet}
\renewcommand{\u}{\circ}
\renewcommand{\a}{b}

\newcommand{\ld}[1]{\langle #1 \rangle}
\newcommand{\lb}[1]{[#1]}

\begin{document}
\title{Focus-style proofs for the two-way alternation-free $\mu$-calculus}
%
%
\author{Jan Rooduijn\thanks{The research of this author has been made possible by a grant from the Dutch Research Council NWO, project number 617.001.857.} \and
Yde Venema}
\authorrunning{J.M.W. Rooduijn and Y. Venema}
%
\institute{ILLC, University of Amsterdam, The Netherlands}
\maketitle              
\begin{abstract}
We introduce a cyclic proof system for the two-way alterna-tion-free modal $\mu$-calculus. The system manipulates one-sided Gentzen sequents and locally deals with the backwards modalities by allowing analytic applications of the cut rule. The global effect of backwards modalities on traces is handled by making the semantics relative to a specific strategy of the opponent in the evaluation game. This allows us to augment sequents by so-called trace atoms, describing traces that the proponent can construct against the opponent's strategy. The idea for trace atoms comes from Vardi's reduction of alternating two-way automata to deterministic one-way automata. Using the multi-focus annotations introduced earlier by Marti and Venema, we turn this trace-based system into a path-based system. We prove that our system is sound for all sequents and complete for sequents not containing trace atoms.  

\end{abstract}
\section{Introduction}
The modal $\mu$-calculus, introduced in its present form by Kozen~\cite{Kozen83}, is an extension of modal logic by least and greatest fixed point operators. It retains many of the desirable properties of modal logic, such as bisimulation invariance, and relatively low complexity of the model-checking and satisfiability problems. Nevertheless, the modal $\mu$-calculus achieves a great gain in expressive power, as the fixed point operators can be used to capture a form of recursive reasoning. This is illustrated by the fact that the modal $\mu$-calculus embeds many well-known extensions of modal logic, such as Common Knowledge Logic, Linear Temporal Logic and Propositional Dynamic Logic. 

A natural further extension is to add a converse modality $\breve a$ for each modality $a$. The resulting logic, called \emph{two-way modal $\mu$-calculus}, can be viewed as being able to reason about the past. As such, it can interpret the past operator of Tense Logic, and moreover subsumes $\mathsf{PDL}$ with converse. In this paper we are concerned with the proof theory of the two-way modal $\mu$-calculus.

Developing good proof systems for the modal $\mu$-calculus is notoriously difficult. In~\cite{Kozen83}, Kozen introduced a natural Hilbert-style axiomatisation, which was proven to be complete only more than a decade later by Walukiewicz~\cite{Walukiewicz00}. Central to this proof is the use of tableau systems introduced by Niwi{\'n}ski and Walukiewicz in~\cite{NiwinskiW96}. One perspective on these tableau systems is that they are cut-free Gentzen-style sequent systems allowing infinite branches. A proof in such a system, called a \emph{non-well-founded proof}, is accepted whenever every infinite branch satisfies a certain progress condition. In case this progress condition is $\omega$-regular (as it is in the case of the modal $\mu$-calculus), automata-theoretic methods show that for every non-well-founded proof there is a \emph{regular} proof, \emph{i.e.} a proof tree containing only finitely many non-isomorphic subtrees. Since these kind of proofs can be naturally presented as finite trees with back edges, they are called \emph{cyclic proofs}. As an alternative to non-well-founded proofs, one can use proof rules with infinitely many premisses. We will not take this route, but note that it has been applied to the two-way modal $\mu$-calculus by Afshari, J\"ager and Leigh in~\cite{AfshariJL19}. 

In~\cite{LangeS01} Lange and Stirling, for the logics $\mathsf{LTL}$ and $\mathsf{CTL}$, annotate formulas in sequents with certain automata-theoretic information. This makes it possible to directly construct cyclic proof systems, without the detour through automata theory. This technique has been further developed by Jungteerapanich and Stirling~\cite{Stirling14,Jungteerapanich09} for the modal $\mu$-calculus. Moreover, certain fragments of the modal $\mu$-calculus, such as the alternation-free fragment~\cite{MartiVenema21} and modal logic with the master modality~\cite{Rooduijn21} have received the same treatment. Encoding automata-theoretic information in cyclic proofs, through annotating formulas, makes them more amenable to proof-theoretic applications, such as the extraction of interpolants from proofs~\cite{MartiVenema21arxiv,AfshariLT21}. 

The logic at hand, the two-way modal $\mu$-calculus, poses additional difficulties. Already without fixed point operators, backwards modalities are known to require more expressivity than offered by a cut-free Gentzen system~\cite{Ohnishi1957}. A common solution is to add more structure to sequents, as \emph{e.g.} the nested sequents of Kashima~\cite{Kashima94}. This approach, however, does not combine well with cyclic proofs, as the number of possible sequents in a given proof becomes unbounded. We therefore opt for the alternative approach of still using ordinary sequents, but allowing analytic applications of the cut rule (see~\cite{Gore1999} for more on the history of this approach). The combination of analytic cuts and cyclic proofs has already been shown to work well in the case of Common Knowledge Logic~\cite{Zenger22}. Choosing analytic cuts over sequents with extended structure has recently also been gaining interest in the proof theory of logics without fixed point operators~\cite{ciabattoni2022theory}.

Although allowing analytic cuts handles the backwards modalities on a local level, further issues arise on a global level in the combination with non-well-founded branches. The main challenge is that the progress condition should not just hold on infinite branches, but also on paths that can be constructed by moving both up and down a proof tree. Our solution takes inspiration from Vardi's reduction of alternating two-way automata to deterministic one-way automata~\cite{Vardi98}. Roughly, the idea is to view these paths simply as upwards paths, only interrupted by several detours, each returning to the same state as where it departed. One of the main insights of the present research is that such detours have a natural interpretation in terms of the game semantics of the modal $\mu$-calculus. We exploit this by extending the syntax with so-called \emph{trace atoms}, whose semantics corresponds with this interpretation. Our sequents will then be one-sided Gentzen sequents containing annotated formulas, trace atoms, and negations of trace atoms. 

For the sake of simplicity we will restrict ourselves to the alternation-free fragment of the modal $\mu$-calculus. This roughly means that we will allow no entanglement of least and greatest fixed point operators. In this setting it suffices to annotate formulas with just a single bit of information, distinguishing whether the formula is \emph{in focus}~\cite{MartiVenema21}. This is a great simplification compared to the full language, where annotations need to be strings and a further global annotation, called the \emph{control}, is often used~\cite{Stirling14,Jungteerapanich09}. Despite admitting simple annotations, the trace structure of the alternation-free modal $\mu$-calculus remains intricate. This is mainly caused by the fact that disjunctions may still appear in the scope of greatest fixed point operators, causing traces to split. 

While this paper was under review, the preprint \cite{Enqvistetal} by Enqvist et al. appeared, in which a proof system is presented for the two-way modal $\mu$-calculus (with alternation). Like our system, their system is cyclic. Moreover, they also extend the syntax in order to apply the techniques from Vardi in a proof-theoretical setting. However, their extension, which uses so-called \emph{ordinal variables}, is substantially different from ours, which uses trace atoms. It would be interesting to see whether the two approaches are intertranslatable. 

In Section 2 we define the two-way alternation-free modal $\mu$-calculus. Section 3 is devoted to introducing the proof system, after which in Section 4 we show that proofs correspond to winning strategies in a certain parity game. In Section 5 we prove soundness and completeness. The concluding Section 6 contains a short summary and some ideas for further research. 
\section{The (alternation-free) two-way modal $\mu$-calculus}
For the rest of this paper we fix the countably infinite sets $\mathsf{P}$ of \emph{propositional variables} and $\mathsf{D}$ of \emph{actions}. Since we want our modal logic to be \emph{two-way}, we define an involution operation $\breve \cdot : \mathsf{D} \rightarrow \mathsf{D}$ such that for every $a \in \mathsf{D}$ it holds that $\breve a \not= a$ and $\breve{\breve{a}} = a$.
We work in negation normal form, where the language $\mathcal{L}_{2\mu}$ of the \emph{two-way modal $\mu$-calculus} is generated by the following grammar:
\[
\varphi ::= p \mid \overline p \mid \varphi \lor \psi \mid \varphi \land \psi \mid \ld{a} \varphi \mid \lb{a} \varphi \mid \mu x \varphi \mid \nu x \varphi
\]
where $p, x \in \mathsf{P}$, $a \in \mathsf{D}$ and in the formation of $\eta x \varphi$ ($\eta \in \{\mu, \nu\})$ the formula $\overline x$ does not occur in $\varphi$. The language $\mathcal{L}_{2 \mu}$ expresses $\top$ and $\bot$, \emph{e.g.} as $\nu x .x$ and $\mu x .x$. For the reader familiar with the ordinary modal $\mu$-calculus, note that the only distinctive feauture of $\mathcal{L}_{2 \mu}$ is the assumed involution operator on $\mathsf{D}$. 

We use standard terminology for the binding of variables by a fixpoint operator $\eta$. In particular, we write $FV(\varphi)$ for the set of variables $x \in \mathsf{P}$ that occur freely in $\varphi$ and $BV(\varphi)$ for the set of those that are bound by some fixpoint operator. Note that for every $\overline x$ occurring in $\varphi$, we have $x \in FV(\varphi)$. For technical convenience, we assume that each formula $\varphi$ is \emph{tidy}, \emph{i.e.} that $FV(\varphi) \cap BV(\varphi) = \emptyset$. The \emph{unfolding} of a formula $\psi = \eta x \varphi$ is the formula $\varphi[\psi/x]$, obtained by substituting every free occurrence of $x$ in $\varphi$ by $\psi$. No free variables of $\psi$ are captured by this procedure, because $FV(\psi) \cap BV(\varphi) \subseteq FV(\varphi) \cap BV(\varphi) = \emptyset$. The \emph{closure} of a formula $\xi \in \mathcal{L}_{2 \mu}$ is the least set $\mathsf{Clos}(\xi) \subseteq \mathcal{L}_{2 \mu}$ such that $\xi \in \mathsf{Clos}(\xi)$ and:
\begin{enumerate}[label = (\roman*), align= left]
	\item $\varphi \circ \psi \in \mathsf{Clos}(\xi)$ implies $\varphi, \psi \in \mathsf{Clos}(\xi)$ for each $\circ \in \{\lor, \land\}$;
	\item $\triangle \varphi \in \mathsf{Clos}(\xi)$ implies $\varphi \in \mathsf{Clos}(\xi)$ for every $\triangle \in \{\langle a \rangle, [a] \mid a \in \mathsf{D}\}$;
	\item $\eta x \varphi \in \mathsf{Clos}(\xi)$ implies $\varphi[\eta x \varphi /x] \in \mathsf{Clos}(\xi)$ for every $\eta \in \{\mu, \nu\}$.  
\end{enumerate}
It is well known that $\mathsf{Clos}(\xi)$ is always finite and that all formulas in $\mathsf{Clos}(\xi)$ are tidy if $\xi$ is so (see \emph{e.g.} \cite{venema2020lectures}). 

Formulas of $\mathcal{L}_{2\mu}$ are interpreted in \emph{Kripke models} $\mathbb{S} = (S, (R_a)_{a \in \mathsf{D}}, V)$, where $S$ is a set of \emph{states}, for each $a \in \mathsf{D}$ we have an \emph{accessibility relation} $R_a \subseteq S \times S$, and $V : \mathsf{P} \rightarrow \mathcal{P}(S)$ is a \emph{valuation function}. We assume that each model is \emph{regular}, \emph{i.e.} that $R_a$ is the converse relation of $R_{\breve{a}}$ for every $a \in \mathsf{D}$. Recall that the converse relation of a relation $R$ consists of those $(y, x)$ such that $(x, y) \in R$. 

We set $R_a[s] := \{t \in S : sR_at\}$ and let $\mathbb{S}[x \mapsto X]$ be the model obtained from $\mathbb{S}$ by replacing the valuation function $V$ by $V[x \mapsto X]$, defined by setting $V[x \mapsto X](x) = X$ and $V[x \mapsto X](p) = V(p)$ for every $p \not= x$. The \emph{meaning} $\llbracket \varphi \rrbracket^\mathbb{S} \subseteq S$ of a formula $\xi \in \mathcal{L}_{2\mu}$ in $\mathbb{S}$ is inductively on the complexity of $\xi$: 
{\small
\begin{alignat*}{4}
&\llbracket p \rrbracket^\mathbb{S} &&:= V(p) &&\llbracket \overline p \rrbracket^\mathbb{S} &&:= S \setminus V(p) \\
&\llbracket \varphi \lor \psi \rrbracket^\mathbb{S} &&:= \llbracket \varphi \rrbracket^\mathbb{S} \cup \llbracket \psi \rrbracket^\mathbb{S}  && \llbracket \varphi \land \psi \rrbracket^\mathbb{S} &&:= \llbracket \varphi \rrbracket^\mathbb{S} \cap \llbracket \psi \rrbracket^\mathbb{S} \\
&\llbracket \langle a \rangle \varphi \rrbracket^\mathbb{S} &&:= \{s \in S \mid R_a[s] \cap \llbracket \varphi \rrbracket^\mathbb{S} \not= \emptyset\}  && \llbracket [a] \varphi \rrbracket^\mathbb{S} &&:= \{s \in S \mid R_a[s] \subseteq \llbracket \varphi \rrbracket^\mathbb{S}\} \\
&\llbracket \mu x \varphi \rrbracket^\mathbb{S} &&:= \bigcap\{X \subseteq S \mid \llbracket \varphi \rrbracket^{\mathbb{S}[x \mapsto X]} \subseteq X\} \ \ \  && \llbracket \nu x \varphi \rrbracket^\mathbb{S} &&:= \bigcup\{X \subseteq S \mid X \subseteq \llbracket \varphi \rrbracket^{\mathbb{S}[x \mapsto X]}\} 
\end{alignat*}
}We will use the definable (see~\cite{venema2020lectures}) negation operator $\overline{\cdot}$ on $\mathcal{L}_{2 \mu}$, for which it holds that $\llbracket \overline{\xi} \rrbracket^\mathbb{S} = S \setminus \llbracket \xi \rrbracket^\mathbb{S}$. 

In this paper we shall only work with an alternative, equivalent, definition of the semantics, given by the \emph{evaluation game} $\mathcal{E}(\xi, \mathbb{S})$. We refer the reader to the appendix below for the basic notions of (parity) games. The game $\mathcal{E}(\xi, \mathbb{S})$ is played on the board $\mathsf{Clos}(\xi) \times S$, and its ownership function and admissible moves are given in the following table. 
	\begin{center}
		\begin{tabular}{|c|c|c|}
			\hline
			Position & Owner & Admissible moves \\
			\hline
			$(p, s), s \in V(p)$ & $\forall$ & $\emptyset$\\
			$(p, s), s \notin V(p)$ & $\exists$ & $\emptyset$\\
			$(\varphi \lor \psi, s)$ & $\exists$ & $\{(\varphi, s), (\psi, s)\}$ \\
			$(\varphi \land \psi, s)$ & $\forall$ & $\{(\varphi, s), (\psi, s)\}$ \\
			$(\langle a \rangle \varphi, s)$ & $\exists$ & $\{\varphi\} \times R_a[s]$ \\	
			$([a]\varphi, s)$ & $\forall$ & $\{\varphi\} \times R_a[s]$ \\	
			$(\eta x \varphi, s)$ & $-$ & $\{(\varphi [\eta x\varphi /  x], s)\}$ \\
			\hline
		\end{tabular}
	\end{center}
The following proposition is standard in the literature on the modal $\mu$-calculus. See~\cite[Proposition 6.7]{kupke2020size} for a proof. 
\begin{proposition}
For every infinite $\mathcal{E}(\xi, \mathbb{S})$-match $\mathcal{M} = (\varphi_n, s_n)_{n \in \omega}$, there is a unique fixpoint formula $\eta x \chi$ which occurs infinitely often in $\mathcal{M}$ and is a subformula of $\varphi_n$ for cofinitely many $n$.
\end{proposition}
The winner of an infinite match $\mathcal{E}(\xi, \mathbb{S})$-match is $\exists$ if in the previous proposition $\eta = \nu$, and $\forall$ if $\eta = \mu$. It is well known that $\mathcal{E}(\xi, \mathbb{S})$ can be realised as a parity game by defining a suitable priority function on $\mathsf{Clos}(\xi) \times S$ (we again refer the reader to~\cite{kupke2020size} for a detailed proof of this fact). Because of this we may, by Theorem \ref{thm:posdet} in Appendix \ref{sec:apppar}, assume that winning strategies are optimal and positional. Finally, we state the known fact that the two approaches provide the same meaning to formulas. For every $\varphi \in \mathsf{Clos}(\xi)$: $(\varphi, s) \in \text{Win}_\exists(\mathcal{E}(\xi, \mathbb{S}))@(\varphi, s)$ if and only if $s \in \llbracket \varphi \rrbracket^\mathbb{S}$. If either is side of the bi-implication holds, we say that $\varphi$ \emph{is satisfied} in $\mathbb{S}$ at $s$ and write $\mathbb{S}, s \Vdash \varphi$. 

In this paper we are concerned with a fragment of $\mathcal{L}_{2 \mu}$ containing only those formulas $\xi$ which are \emph{alternation free}, \emph{i.e.} such that for every subformula $\eta x \varphi$ of $\xi$ it holds that no free occurrence of $x$ in $\varphi$ is in the scope of an $\overline \eta$-operator in $\varphi$ (where $\overline \eta$ denotes the opposite fixed point operator of $\eta$). This fragment is called \emph{the alternation-free two-way modal $\mu$-calculus} and denoted by $\mathcal{L}_{2 \mu}^{\textit{af}}$. We close this section by stating some typical properties of the alternation-free fragment. For $\eta \in \{\mu, \nu\}$ we use the term \emph{$\eta$-formula} for a formula of the form $\eta x \varphi$. 
\begin{proposition}
\label{prop:propertiesaf}
Let $\xi \in \mathcal{L}^{\textit{af}}_{2 \mu}$ be an alternation-free formula. Then:
\begin{itemize}[label = $\bullet$]
	\item Every formula $\varphi \in \mathsf{Clos}(\xi)$ is alternation free. 
	\item The negation $\overline \xi$ is alternation free. 
	\item An infinite $\mathcal{E}(\xi, \mathbb{S})$-match is won by $\exists$ precisely if it contains infinitely many $\nu$-formulas, and by $\forall$ precisely if it contains infinitely many $\mu$-formulas. 
\end{itemize}
\end{proposition}
\section{The proof system}
We will call a set $\Sigma$ of formulas \emph{negation-closed} if for every $\xi \in \Sigma$ it holds that $\overline \xi \in \Sigma$ and $\mathsf{Clos}(\xi) \subseteq \Sigma$. For the remainder of this paper we fix a finite and negation-closed set $\Sigma$ of $\mathcal{L}_{2 \mu}^\textit{af}$-formulas. For reasons of technical convenience, we will assume that every formula is drawn from $\Sigma$. This does not restrict the scope of our results, as any formula is contained in some finite negation-closed set. 
\subsection{Sequents}
\subsubsection{Syntax}
Inspired by~\cite{MartiVenema21}, we annotate formulas by a single bit of information.
\begin{definition}
An \emph{annotated formula} is a formula with an annotation in $\{\u, \f\}$.
\end{definition}
The letters $b, c, d, \ldots$ are used as variables ranging over the annotations $\u$ and $\f$. An annotated formula $\varphi^b$ is said to be \emph{out of focus} if $b = \u$, and \emph{in focus} if $b = \f$. The focus annotations will keep track of so-called \emph{traces} on paths through proofs. Roughly, a trace on a path is a sequence of formulas, such that the $i$-th formula occurs in the $i$-th sequent on the path, and the $i + 1$-th formula `comes from' the $i$-th formula in a way which we will define later. In Section \ref{sec:proofsearchgame} we will construct a game in which the winning strategies of one player correspond precisely to the proofs in our proof system. The focus mechanism enables us to formulate this game as a parity game. This is essentially also the approach taken in~\cite{MartiVenema21}. 
 
Where traces usually only moves upwards in a proof, the backwards modalities of our language will be enable them to go downwards as well. We will handle this in our proof system by further enriching our sequents with the following additional information. 
\begin{definition}
For any two formulas $\varphi, \psi$, there is a \emph{trace atom} $\varphi \leadsto \psi$ and a \emph{negated trace atom} $\varphi \not \leadsto \psi$. 
\end{definition}
The idea for trace atoms will become more clear later, but for now one can think of $\varphi \leadsto \psi$ as expressing that there is some kind of trace going from $\varphi$ to $\psi$, and of $\varphi \not \leadsto \psi$ as its negation. Finally, our sequents are built from the above three entities. 
\begin{definition}
A \emph{sequent} is a finite set consisting of annotated formulas, trace atoms, and negated trace atoms.
\end{definition}
Whenever we want to refer to general elements of a sequent $\Gamma$, without specifying whether we mean annotated formulas or (negated) trace atoms, we will use the capital letters $A, B, C, \ldots$. 
\subsubsection{Semantics} We will now define the semantics of sequents. Unlike annotations, which do not affect the semantics but only serve as bookkeeping devices, the trace atoms have a well-defined interpretation. We will work with a refinement of the usual satisfaction relation that is defined with respect to a strategy for $\forall$ in the evaluation game. Most of the time, this strategy will be both \emph{optimal} and \emph{positional} (see Appendix \ref{sec:apppar} for the precise definition of these terms). Because we will frequently need to mention such optimal positional strategies, we will refer to them by the abbreviation \emph{ops}. We first define the interpretation of annotated formulas. Note that the focus annotations play no role in this definition. 
\begin{definition}
Let $\mathbb{S}$ be a model, let $f$ be an ops for $\forall$ in $\mathcal{E}@(\bigwedge \Sigma, \mathbb{S})$ and let $\varphi^b$ be an annotated formula. We write $\mathbb{S}, s \Vdash_f \varphi^b$ if $f$ is \emph{not} winning for $\forall$ at $(\varphi, s)$.
\end{definition}
The following proposition, which is an immediate consequence of Theorem \ref{thm:posdet} of the appendix, relates $\Vdash_f$ to the usual satisfaction relation $\Vdash$.
\begin{proposition}
$\mathbb{S}, s \Vdash \varphi$ iff for every ops $f$ for $\forall$ in $\mathcal{E}(\bigwedge \Sigma, \mathbb{S})$: $\mathbb{S}, s \Vdash_f \varphi^b$.
\end{proposition}
The semantics of trace atoms is also given relative to an ops for $\forall$ in the game $\mathcal{E}(\bigwedge \Sigma, \mathbb{S})$ (in the following often abbreviated to $\mathcal{E}$). 
\begin{definition}
Given an ops $f$ for $\forall$ in $\mathcal{E}$, we say that $\varphi \leadsto \psi$ is \emph{satisfied} in $\mathbb{S}$ at $s$ with respect to $f$ (and write $\mathbb{S}, s \Vdash_f \varphi \leadsto \psi$) if there is an $f$-guided match
\[
 (\varphi, s) = (\varphi_0, s_0) \cdot (\varphi_1, s_1) \cdots (\varphi_n, s_n) = (\psi, s) \ \ \ (n \geq 0)
\]
such that for no $i < n$ the formula $\varphi_i$ is a $\mu$-formula. We say that $\mathbb{S}$ \emph{satisfies} $\varphi \not \leadsto \psi$ at $s$ with respect to $f$ (and write $\mathbb{S}, s \Vdash_f \varphi \not \leadsto \psi$) iff $\mathbb{S}, s \not \Vdash_f \varphi \leadsto \psi$.
\end{definition}
The idea behind the satisfaction of a trace atom $\varphi \leadsto \psi$ at a state $s$ is that $\exists$ can take the match from $(\varphi, s)$ to $(\psi, s)$ without passing through a $\mu$-formula. This is good for the player $\exists$. For instance, if $\varphi \leadsto \psi$ and $\psi \leadsto \varphi$ are satisfied at $s$ with respect to $f$ for some $\varphi \not= \psi$, then $f$ is necessarily losing for $\forall$ at the position $(\varphi, s)$. We will later relate trace atoms to traces in infinitary proofs.

We interpret sequents disjunctively, that is: $\mathbb{S}, s \Vdash_f \Gamma$ whenever $\mathbb{S}, s \Vdash_f A$ for some $A \in \Gamma$. The sequent $\Gamma$ is said to be \emph{valid} whenever $\mathbb{S}, s \Vdash_f \Gamma$ for every model $\mathbb{S}$, state $s$ of $\mathbb{S}$, and ops $f$ for $\forall$ in $\mathcal{E}$. 
\begin{remark}
There is another way in which one could interpret sequents, which corresponds to what one might call \emph{strong validity}, and which the reader should note is different from our notion of validity. Spelling it out, we say that $\Gamma$ is \emph{strongly valid} if for every model $\mathbb{S}$ and state $s$ there is an $A$ in $\Gamma$ that such that for every ops $f$ for $\forall$ in $\mathcal{E}$ it holds that $\mathbb{S}, s \Vdash_f A$. While these two notions coincide for sequents containing only annotated formulas, the sequent given by $\{\varphi \land \psi \leadsto \varphi, \varphi \land \psi \leadsto \psi$\} shows that they do not in general. 
\end{remark}
We finish this subsection by defining three operations on sequents that, respectively, extract the formulas contained annotated in some sequent, take all annotated formulas out of focus, and put all formulas into focus.
\begin{align*}
&\Gamma^- &&:= \{\chi \mid \chi^{\a}  \in \Gamma \text{ for some $b \in \{\u, \f\}$}\}, \\
&\Gamma^\u &&:= \{\varphi \leadsto \psi \mid \varphi \leadsto \psi \in \Gamma\} \cup \{\varphi \not \leadsto \psi \mid \varphi \not \leadsto \psi \in \Gamma\} \cup \{\chi^\u \mid \chi  \in \Gamma^-\}, \\
&\Gamma^\f &&:= \{\varphi \leadsto \psi \mid \varphi \leadsto \psi \in \Gamma\} \cup \{\varphi \not \leadsto \psi \mid \varphi \not \leadsto \psi \in \Gamma\} \cup \{\chi^\f \mid \chi  \in \Gamma^-\}.
\end{align*}
\subsection{Proofs}
In this subsection we give the rules of our proof system. Because the rule for modalities is quite involved, its details are given in a separate definition.
\begin{definition}
\label{defn:jump}
Let $\Gamma$ be a sequent and let $[a]\varphi^{b}$ be an annotated formula. The \emph{jump} $\Gamma^{[a] \varphi^{b}}$ of $\Gamma$ with respect to $[a] \varphi^{b}$ consists of: 
\begin{enumerate}
	\item \begin{enumerate}
		\item $\varphi^{s([a]\varphi, \Gamma)}$;
		\item $\psi^{s(\langle a \rangle \psi, \Gamma)}$ for every $\langle a \rangle \psi^{c} \in \Gamma$;
		\item $[\breve{a}] \chi^{\u}$ for every $\chi^{d} \in \Gamma$ such that $[\breve{a}]\chi \in \Sigma$; \\
	\end{enumerate}
	\item \begin{enumerate}
		\item $\varphi \leadsto \langle \breve{a} \rangle \chi$ for every $[a]\varphi \leadsto \chi \in \Gamma$ \hspace{0.0000001pt} such that $\langle \breve{a} \rangle  \chi \in \Sigma$;
		\item $\langle \breve a \rangle \chi \not \leadsto \varphi$ for every $\chi \not \leadsto [a]\varphi \in \Gamma$  \hspace{0.0000000001pt}  such that $\langle \breve a \rangle \chi \in \Sigma$;
		\item $\psi \leadsto \langle \breve{a} \rangle  \chi$ for every $\langle a \rangle  \psi \leadsto \chi \in \Gamma$ such that $\langle \breve{a} \rangle  \chi \in \Sigma$;
		\item $\langle \breve a \rangle \chi \not \leadsto \psi$ for every $\chi \not \leadsto \langle a \rangle \psi \in \Gamma$  such that $\langle \breve a \rangle \chi \in \Sigma$, 
		\end{enumerate}
\end{enumerate}
where $s(\xi, \Gamma)$ is defined by:
\[
{s(\xi, \Gamma)} = \begin{cases}
\f &\text{if } \xi^{\f} \in \Gamma, \\
\f &\text{if } \text{$\theta \not \leadsto \xi \in \Gamma$ for some $\theta^{\f} \in \Gamma$,} \\
\u &\text{otherwise.}
\end{cases}
\]
\end{definition}
Before we go on to provide the rest of the proof system, we will give some intuition for the modal rule, by proving the lemma below. This lemma essentially expresses that the modal rule is sound. Since the annotations play no role in the soundness of an individual rule, we suppress the annotations in the proof below for the sake of readability. Intuition for the annotations in the modal rule, and in particular for the function $s$, is given later. 
\begin{lemma}
\label{lem:soundnessmodalrule}
Given a model $\mathbb{S}$, a state $s$ of $\mathbb{S}$, and an ops $f$ for $\forall$ in $\mathcal{E}$ such that $\mathbb{S}, s \not \Vdash_f [a]\varphi^b, \Gamma$, there is an $a$-successor $t$ of $s$, such that $\mathbb{S}, t \not \Vdash_f \Gamma^{[a]\varphi^b}$.
\end{lemma}
\begin{proof}
Let $t$ be the state chosen by $f([a]\varphi, s)$. We claim that $\mathbb{S}, t \not \Vdash_f \Gamma^{[a] \varphi^b}$. To start with, since $f$ is winning, we have $\mathbb{S}, t \not \Vdash_f \varphi$. Moreover, if $\langle a \rangle \psi^c$ belongs to $\Gamma$, then $\mathbb{S}, s \not \Vdash_f \langle a \rangle \psi$ and thus $\mathbb{S}, t \not \Vdash_f \psi$. Thirdly, if $\chi^d$ belongs to $\Gamma$ and $\lb{\breve{a}} \chi \in \Sigma$, then $\mathbb{S}, s \not \Vdash [\breve {a}] \chi$, whence by the optimality of $f$, we have $\mathbb{S}, t \not \Vdash_f \lb{\breve{a}} \chi$. 

The above shows all conditions under item 1. For the conditions under item 2, suppose that $\ld{\breve{a}} \chi \in \Sigma$. We only show 2(d), because the others are similar. Suppose that $\chi \not \leadsto \ld{a} \psi \in \Gamma$. Then $\mathbb{S}, s \not \Vdash_f \chi \not \leadsto \ld{a} \psi$, whence $\mathbb{S}, s \Vdash_f \chi \leadsto \ld{a} \psi$. That means that there is an $f$-guided $\mathcal{E}$-match
\[
 (\chi, s) = (\varphi_0, s_0) \cdot (\varphi_1, s_1) \cdots (\varphi_n, s_n) = (\ld{a} \psi, s) \ \ \ (n \geq 0)
\]
such that none of the $\varphi_i$'s is a $\mu$-formula. But then the $f$-guided $\mathcal{E}$-match
\[
(\ld{\breve{a}} \chi, t) \cdot (\varphi_0, s_0) \cdots (\varphi_n, s_n) \cdot (\psi, t)
\]
witnesses that $\mathbb{S}, t \not \Vdash_f \ld{\breve{a}} \chi \not \leadsto \psi$, as required. 
\end{proof}
	\begin{figure}
	\FrameSep-6pt
	\begin{framed}
	\small
	\begin{align*}
	\AxiomC{}
	\RightLabel{$\mathsf{Ax1}$}		
	\UnaryInfC{$\varphi^b, \overline \varphi^c, \Gamma$}
	\DisplayProof
	&&
	\AxiomC{}
	\RightLabel{$\mathsf{Ax2}$}		
	\UnaryInfC{$\varphi \leadsto \psi, \varphi \not \leadsto \psi, \Gamma$}
	\DisplayProof
	&&
	\AxiomC{}
	\RightLabel{$\mathsf{Ax3}$}		
	\UnaryInfC{$\varphi \leadsto \varphi, \Gamma$}
	\DisplayProof
	\end{align*}
	\begin{align*}
	\AxiomC{$(\varphi \lor \psi) \not \leadsto \varphi, (\varphi \lor \psi) \not \leadsto \psi, \varphi^{\a}, \psi^{\a}, \Gamma$}
	\RightLabel{$\mathsf{R}_\lor$}
	\UnaryInfC{$\varphi \lor \psi^{\a}, \Gamma$}
	\DisplayProof
	&&
	\AxiomC{$\varphi^{\u}, \Gamma$}
	\AxiomC{$\overline{\varphi}^{\u}, \Gamma$}
	\RightLabel{$\mathsf{cut}$}
	\BinaryInfC{$\Gamma$}
	\DisplayProof
	\end{align*}
	\begin{align*}
	\AxiomC{$(\varphi \land \psi) \not \leadsto \varphi, \varphi^{\a}, \Gamma$}
	\AxiomC{$(\varphi \land \psi) \not \leadsto \psi, \psi^{\a}, \Gamma$}
	\RightLabel{$\mathsf{R}_\land$}
	\BinaryInfC{$\varphi \land \psi^{\a}, \Gamma$}
	\DisplayProof
	&&
	\AxiomC{$\varphi[\mu x \varphi/x]^{\u}, \Gamma$}
	\RightLabel{$\mathsf{R}_{\mu}$}
	\UnaryInfC{$\mu x \varphi^{\a}, \Gamma$}
	\DisplayProof
\end{align*}
\begin{align*}
	\AxiomC{$\nu x \varphi \not \leadsto \varphi[\nu x \varphi/x],\varphi[\nu x \varphi / x] \leadsto \nu x \varphi, \varphi[\nu x \varphi/x]^{\a}, \Gamma$}
	\RightLabel{$\mathsf{R}_\nu$}
	\UnaryInfC{$\nu x \varphi^{\a}, \Gamma$}
	\DisplayProof
	&&
	\AxiomC{$\Gamma^{[a]\varphi^{\a}}$}
	\RightLabel{$\mathsf{R}_{[a]}$}
	\UnaryInfC{$[a]\varphi^{\a}, \Gamma$}
	\DisplayProof
	\end{align*}	
	\begin{align*}
	\AxiomC{$\Gamma^{\f}$}
	\RightLabel{$\mathsf{F}$}
	\UnaryInfC{$\Gamma^{\u}$}
	\DisplayProof
	&&
	\AxiomC{$\varphi \not \leadsto \psi, \psi \not \leadsto \chi, \varphi \not \leadsto \chi, \Gamma$}
	\RightLabel{$\mathsf{trans}$}
	\UnaryInfC{$\varphi \not \leadsto \psi, \psi \not \leadsto \chi, \Gamma$}
	\DisplayProof
	&&
	\AxiomC{$\varphi \leadsto \psi, \Gamma$}
	\AxiomC{$\varphi \not \leadsto \psi, \Gamma$}
	\RightLabel{$\mathsf{tc}$}
	\BinaryInfC{$\Gamma$}
	\DisplayProof
	\end{align*}
  	\end{framed}
	\caption{The proof rules of the system $\mathsf{Focus^2}$.}
	\label{fig:rules}
	\end{figure}
	The rules of the system $\mathsf{Focus^2}$ are given in Figure \ref{fig:rules}. In each rule, the annotated formulas occurring in the set $\Gamma$ are called \emph{side formulas}. Moreover, the rules in $\{\mathsf{R}_\lor, \mathsf{R}_\land, \mathsf{R}_\mu, \mathsf{R}_\nu, \mathsf{R}_{\lb{a}}\}$ have precisely one \emph{principal formula}, which by definition is the annotated formula appearing to the left of $\Gamma$ in the conclusion. Note that, due to the fact that sequents are taken to be sets, an annotated formula may at the same time be both a principal formula and a side formula.

	We will now define the relation of \emph{immediate ancestry} between formulas in the conclusion and formulas in the premisses of some arbitrary rule application. For any side formula in the conclusion of some rule, we let its \emph{immediate ancestors} be the corresponding side formulas in the premisses. For every rule except $\mathsf{R}_{[a]}$, if some formula in the conclusion is a principal formula, its \emph{immediate ancestors} are the annotated formulas occurring to the left of $\Gamma$ in the premisses. Finally, for the \emph{modal rule} $\mathsf{R}_{[a]}$, we stipulate that $\varphi^{s([a]\varphi, \Gamma)}$ is an \emph{immediate ancestor} of the principal formula $[a]\varphi^{\a}$, and that each $\psi^{s(\langle a \rangle \psi, \Gamma)}$ contained in $\Gamma^{[a] \varphi^{\a}}$ due to clause 1(b) of Definition \ref{defn:jump} is an \emph{immediate ancestor} of $\langle a \rangle \psi^{\a} \in \Gamma$. 

	As mentioned before, the purpose of the focus annotations is to keep track of \emph{traces} of formulas on branches. Usually, a trace is a sequence of formulas $(\varphi_n)_{n < \omega}$ such that each $\varphi_k$ is an immediate ancestor of $\varphi_{k+1}$. The idea is then that whenever an infinite branch has cofinitely many sequents with a formula in focus, this branch contains a trace on which infinitely many formulas are $\nu$-formulas. Disregarding the backwards modalities for now, this can be seen as follows. As long as the focus rule is not applied, any focussed formula is an immediate ancestor of some earlier focussed formula. Since the principal formula of $\mathsf{R}_\mu$ loses focus, while the principal formula of $\mathsf{R}_\nu$ preserves focus, a straightforward application of K\H{o}nig's Lemma shows that every infinite branch contains a trace with infinitely many $\nu$-formulas. We refer the reader to \cite{MartiVenema21} for more details. 
	
	Our setting is slightly more complicated, because the function $s$ in Definition \ref{defn:jump} additionally allows the focus to transfer along negated trace atoms, rather than just from a formula to one of its immediate ancestors. This is inspired by~\cite{Vardi98}, as are the conditions in the second part of Definition \ref{defn:jump}. The main idea is that, because of the backwards modalities, traces may move not only up, but also down a proof tree. To get a grip on these more complex traces, we cut them up in segments consisting of upward paths, which are the same as ordinary traces, and loops, which are captured by the negated trace atoms. This intuitive idea will become explicit in the proof of completeness in Section \ref{sec:soundnesscompleteness}. 

  We are now ready to define a notion of infinitary proofs in $\mathsf{Focus^2}$. 
	\begin{definition}{}
	\label{defn:focus2infty}
	A \emph{$\mathsf{Focus_\infty^2}$-proof} is a (possibly infinite) derivation in $\mathsf{Focus^2}$ with:
	\begin{enumerate}
		\item All leaves are axioms.
		\item On every infinite branch cofinitely many sequents have a formula in focus.
		\item Every infinite branch has infinitely many applications of $\mathsf{R}_{[a]}$.
	\end{enumerate}
	\end{definition}
	As mentioned above, conditions 2 and 3 are meant to ensure that every infinite trace contains infinitely many $\nu$-formulas. We will use this in Section \ref{sec:soundnesscompleteness} to show that infinitary proofs are sound. The key idea is to relate the traces in a proof to matches in a purported countermodel of its conclusion. 
	
	We leave it to the reader to verify that each rule, apart from the modal rule, is truth-preserving with respect to a given model $\mathbb{S}$, state $s$ of $\mathbb{S}$, and ops $f$ for Refuter in $\mathcal{E}(\bigwedge \Sigma, \mathbb{S})$. Since Lemma \ref{lem:soundnessmodalrule} already showed the soundness of the modal rule, we obtain: 
	\begin{proposition}
	Well-founded $\mathsf{Focus_\infty^2}$-proofs are sound.
	\end{proposition}
	We close this section with two examples of $\mathsf{Focus^2_\infty}$-proofs. The first example demonstrates $\mathsf{cut}$ and item 1(c) of Definition \ref{defn:jump}. The second example demonstrates trace atoms. {}
	\begin{example}
		Define the following two formulas:
		\begin{align*}
			\varphi := \mu x (\langle \breve{a} \rangle x \lor p), &&
			\psi := \nu y ([a] x \land \varphi).
		\end{align*}
		The formula $\varphi$ expresses `there is a backwards $a$-path to some state where $p$ holds'. The formula $\psi$ expresses `$\varphi$ holds at every state reachable by a forwards $a$-path'. As our context $\Sigma$ we take least negation-closed set containing $\varphi$ and $\psi$:
		\[
		\{\varphi, \ld{\breve{a}} \varphi \lor p, \ld{\breve{a}} \varphi, p, \psi, \lb{a} \psi \land \varphi, \lb{a} \psi, \overline \varphi, \lb{\breve{a}} \overline \varphi \land \overline p, p, \lb{\breve{a}} \overline \varphi, \overline \psi, \ld{a} \overline \psi \lor \overline \varphi, \ld{a} \overline \psi\}.
		\]
		The implication $p \rightarrow \psi$ is valid, and below we give a $\mathsf{Focus^2_\infty}$-proof. As this particular proof does not rely on trace atoms, we omit them for readability. 
		\begin{align*}
			\AxiomC{}
			\RightLabel{$\mathsf{Ax1}$}
			\UnaryInfC{$\overline p^\f, \psi^\f, \ld{\breve{a}} \varphi^\u, p^\u$}
			\RightLabel{$\mathsf{R}_\lor$}
			\UnaryInfC{$\overline p^\f, \psi^\f, \ld{\breve{a}} \varphi \lor p^\u$}
			\RightLabel{$\mathsf{R}_\mu$}
			\UnaryInfC{$\overline p^\f, \psi^\f, \varphi^\u$}
			\AxiomC{$\pi$}
			\noLine
			\UnaryInfC{$\psi^\f, \lb{\breve{a}} \overline{\varphi}^\u$}
			\RightLabel{$\mathsf{R}_{[a]}$}
			\UnaryInfC{$\overline{p}^\f, \lb{a} \psi^\f, \overline{\varphi}^\u$}
			\AxiomC{}
			\RightLabel{$\mathsf{Ax1}$}
			\UnaryInfC{$\overline{p}^\f, \varphi^\f, \overline{\varphi}^\u$}
			\RightLabel{$\mathsf{R}_\land$}
			\BinaryInfC{$\overline{p}^\f, \lb{a} \psi \land \varphi^\f, \overline{\varphi}^\u$}
			\RightLabel{$\mathsf{R}_\nu$}
			\UnaryInfC{$\overline{p}^\f, \psi^\f, \overline{\varphi}^\u$}
			\RightLabel{$\mathsf{cut}$}
			\BinaryInfC{$\overline p^\f, \psi^\f$}
			\DisplayProof
		\end{align*}
		In the above proof, the proof $\pi$ is given by
		\begin{align*}
			\AxiomC{}
			\RightLabel{$\mathsf{Ax1}$}
			\UnaryInfC{$\overline{\varphi}^\u, \varphi^\u$}
			\RightLabel{$\mathsf{R}_{\lb{\breve{a}}}$}
			\UnaryInfC{$[a]\psi^\f, \lb{\breve{a}} \overline{\varphi}^\u, \ld{\breve{a}} \varphi^\u, p^\u$}
			\RightLabel{$\mathsf{R}_\lor$}
			\UnaryInfC{$[a]\psi^\f, \lb{\breve{a}} \overline{\varphi}^\u, \ld{\breve{a}} \varphi \lor p^\u$}
			\RightLabel{$\mathsf{R}_{\mu}$}
			\UnaryInfC{$[a]\psi^\f, \lb{\breve{a}} \overline{\varphi}^\u, \varphi^\u$}
			\AxiomC{$\vdots$}
			\noLine
			\UnaryInfC{$\psi^\f, \lb{\breve{a}} \overline{\varphi}^\u$}
			\RightLabel{$\mathsf{R}_{\lb{a}}$}
			\UnaryInfC{$[a]\psi^\f, \lb{\breve{a}} \overline{\varphi}^\u, \overline{\varphi}^\u$}
			\RightLabel{$\mathsf{cut}$}
			\BinaryInfC{$[a]\psi^\f, \lb{\breve{a}} \overline{\varphi}^\u$}
			\AxiomC{}
			\RightLabel{$\mathsf{Ax1}$}
			\UnaryInfC{$\varphi^\u, \overline{\varphi}^\u$}
			\RightLabel{$\mathsf{R}_{\lb{\breve{a}}}$}
			\UnaryInfC{$\ld{\breve{a}} \varphi^\u, p^\u, \lb{\breve{a}} \overline{\varphi}^\u$}
			\RightLabel{$\mathsf{R}_{\lor}$}
			\UnaryInfC{$\ld{\breve{a}} \varphi \lor p^\u, \lb{\breve{a}} \overline{\varphi}^\u$}
			\RightLabel{$\mathsf{R}_{\mu}$}
			\UnaryInfC{$\varphi^\f, \lb{\breve{a}} \overline{\varphi}^\u$}
			\RightLabel{$\mathsf{R}_{\land}$}
			\BinaryInfC{$[a]\psi \land \varphi^\f, \lb{\breve{a}} \overline{\varphi}^\u$}
			\RightLabel{$\mathsf{R}_\nu$}
			\UnaryInfC{$\psi^\f, \lb{\breve{a}} \overline{\varphi}^\u$}
			\DisplayProof
		\end{align*}
		where the vertical dots indicate that the proof continues by repeating what happens at the root of $\pi$. The resulting proof of $\overline p^\f, \psi^\f$ has a single infinite branch, which can easily be seen to satisfy the conditions of Definition \ref{defn:focus2infty}.
	\end{example}
	\begin{example}
	Define $\varphi := \nu x \ld{a} \ld{\breve{a}} x$, \emph{i.e.} $\varphi$ expresses that there is an infinite path of alternating $a$ and $\breve{a}$ transitions. Clearly this holds at every state with an $a$-successor. Hence the implication $\ld{a} p \rightarrow \varphi$ is valid. As context $\Sigma$ we consider the least negation-closed set containing both $\ld{a} p$ and $\varphi$, \emph{i.e.},
	\[
		\{\ld{a} p, p, \varphi, \ld{a} \ld{\breve{a}} \varphi, \ld{\breve{a}} \varphi, [a] \overline p, \overline p, \overline \varphi, \lb{a} \lb{\breve{a}} \overline{\varphi}, \lb{\breve{a}} \overline \varphi \}.
	\]
	The following is a $\mathsf{Focus^2_\infty}$-proof of $\ld{a} p \rightarrow \varphi$.
	\[
	\AxiomC{}
	\RightLabel{$\mathsf{Ax2}$}
	\UnaryInfC{$\overline{p}^\f, \ld{\breve{a}} \varphi^\f, \ld{\breve{a}} \varphi \not \leadsto \ld{\breve{a}} \varphi, \ld{\breve{a}} \varphi \leadsto \ld{\breve{a}} \varphi$}
	\RightLabel{$\mathsf{R}_{\lb{a}}$}
	\UnaryInfC{$\lb{a} \overline{p}^\f, \ld{a} \ld{\breve{a}} \varphi^\f, \varphi \not \leadsto \ld{a} \ld{\breve{a}} \varphi, \ld{a} \ld{\breve{a}} \varphi \leadsto \varphi$}
	\RightLabel{$\mathsf{R}_\nu$}
	\UnaryInfC{$\lb{a} \overline{p}^\f, \varphi^\f$}
	\DisplayProof
	\]
	Note that it is also possible to use $\mathsf{Ax3}$ instead of $\mathsf{Ax2}$ in the above proof.
	\end{example}
	\section{The proof search game}
  \label{sec:proofsearchgame}
	We will define a proof search game $\mathcal{G}(\Sigma)$ for the proof system $\mathsf{Focus_\infty^2}$ in the standard way. First, we require a slightly more formal definition of the notion of a rule instance.
    \begin{definition}
    A \emph{rule instance} is a triple $(\Gamma, \mathsf{r}, \langle \Delta_1, \ldots, \Delta_n \rangle)$ such that 
    \[
    \AxiomC{$\Delta_1 \cdots \Delta_n$}
    \RightLabel{$\mathsf{r}$}
    \UnaryInfC{$\Gamma$}
    \DisplayProof
    \]
    is a valid rule application in $\mathsf{Focus^2}$.
    \end{definition}
    The set of positions of $\mathcal{G}(\Sigma)$ is $\mathsf{Seq}_\Sigma \cup \mathsf{Inst}_\Sigma$, where $\mathsf{Seq}_\Sigma$ is the set of sequents and $\mathsf{Inst}_\Sigma$ is the set of valid rule instances (containing only formulas in $\Sigma$). Since $\Sigma$ is finite, the game $\mathcal{G}(\Sigma)$ has only finitely many positions. The ownership function and admissible moves of $\mathcal{G}(\Sigma)$ are as in the following table:
	\begin{center}
		\begin{tabular}{|c|c|c|}
			\hline
			Position & Owner & Admissible moves \\
			\hline
			$\Gamma \in \mathsf{Seq}_\Sigma$ & \  Prover \  & $\{i \in \mathsf{Inst}_\Sigma \mid \mathsf{conc}(i) = \Gamma\}$ \\
			$(\Gamma, \mathsf{r}, \langle \Delta_1, \ldots, \Delta_n \rangle) \in \mathsf{Inst}_\Sigma$ \  & \ Refuter \ & $\{\Delta_i \mid 1 \leq i \leq n\}$\\
			\hline
		\end{tabular}
	\end{center}
	In the above table, the expression $\mathsf{conc}(i)$ stands for the conclusion (\emph{i.e.} the first element of the triple) of the rule instance $i$. As usual, a finite match is lost by the player who got stuck. An infinite $\mathcal{G}(\Sigma)$-match is won by Prover if and only it has a final segment
	\[
	\Gamma_0 \cdot i_0 \cdot \Gamma_1 \cdot i_1 \cdots
	\]
	on which each $\Gamma_k$ has at least one formula in focus and the instance $i_k$ is an application of $\mathsf{R}_{[a]}$ for infinitely many $k$. The two main observations about $\mathcal{G}(\Sigma)$ that we will use are the following:
	\begin{enumerate}
		\item A $\mathsf{Focus_\infty^2}$-proof of $\Gamma$ is the same as a winning strategy for Prover in $\mathcal{G}(\Sigma)@\Gamma$.
		\item $\mathcal{G}(\Sigma)$ is a parity game, whence positionally determined.
	\end{enumerate}
	The first observation is immediate when viewing a winning strategy as a subtree of the full game tree. To make the second observation more explicit, we give the parity function $\Omega$ for $\mathcal{G}(\Sigma)$. On $\mathsf{Seq}_\Sigma$, we simply set $\Omega(\Gamma) := 0$ for every $\Gamma \in \mathsf{Seq}_\Sigma$. On $\mathsf{Inst}_\Sigma$, we define:
	\[
	\Omega(\Gamma, \mathsf{r}, \langle \Delta_1, \ldots, \Delta_n \rangle) := \begin{cases} 3 & \text{ if $\Gamma$ has no formula in focus,} \\
		2 &\text{ if $\Gamma$ has a formula in focus and $\mathsf{r} = \mathsf{R}_{[a]}$,} \\
		1 &\text{ if $\Gamma$ has a formula in focus and $\mathsf{r} \not= \mathsf{R}_{[a]}$.}
	\end{cases}
	\]
	As a result we immediately obtain a method to reduce general non-well-founded proofs to cyclic proofs. Indeed, if Prover has a winning strategy, she also has \emph{positional} winning strategy, which clearly corresponds to a \emph{regular} $\mathsf{Focus}_{\infty}^\mathsf{2}$-proof (that is, a proof containing only finitely many non-isomorphic subtrees.) 
	\section{Soundness and completeness}
	\label{sec:soundnesscompleteness}
	In this section we will prove the soundness and completeness of the system $\mathsf{Focus_\infty^2}$. More specifically, for soundness we will show that if $\Gamma$ is invalid, then Refuter has a winning strategy in $\mathcal{G}(\Sigma)@\Gamma$. Our completeness result is slightly less wide in scope, showing only that if Refuter has a winning strategy in $\mathcal{G}(\Sigma)@\Gamma$, then $\Gamma^-$ is invalid. 
  \subsection{Soundness}
	For soundness, we assume an ops $f$ for $\forall$ in $\mathcal{E} := \mathcal{E}(\bigwedge \Sigma, \mathbb{S})$ for some $\mathbb{S}$ and $s$ such that $\mathbb{S}, s \not \Vdash_f \Gamma$. The goal is to construct from $f$ a strategy $T_f$ for Refuter in $\mathcal{G} := \mathcal{G}(\Sigma)$. The key idea is to assign to each position $p$ reached in $\mathcal{G}$ a state $s$ such that whenever $p = \Delta \in \mathsf{Seq}_\Sigma$ it holds that $\mathbb{S}, s \not \Vdash_f \Delta$. For $p \in \mathsf{Inst}_\Sigma$, the choice of $T_f$ is then based on $f(\varphi, s)$ where $\varphi$ is a formula determined by the rule instance $p$. The existence of such an $s$ implies that $p$ cannot be an axiom and thus that Refuter never gets stuck. For infinite matches, the proof works by showing that a $T_f$-guided $\mathcal{G}@\Gamma$-match lost by Refuter induces an $f$-guided $\mathcal{E}@\varphi$-match lost by $\forall$. As mentioned above, the key idea here is to relate an $f$-guided $\mathcal{E}@\varphi$-match to a trace through the $T_f$-guided $\mathcal{G}@\Gamma$-match. If the $\mathcal{G}@\Gamma$-match is losing for Refuter, it must contain a trace with infinitely many $\nu$-formulas, which gives us an $\mathcal{E}@\varphi$-match lost by $\forall$. A novel challenge here is that not all steps in a trace necessarily go from a formula to one of its immediate ancestors, but may instead transfer along a negated trace atom. When this happens, say from $\varphi_n$ to $\varphi_{n+1}$, it holds for $\Delta$ as above that both $\varphi_n^\f$ and $\varphi_n \not \leadsto \varphi_{n+1}$ belong to $\Delta$. Since, by the above, it holds that $\mathbb{S}, s \not \Vdash_f \Delta$, we use the fact that $\mathbb{S}, s \Vdash_f \varphi_n \leadsto \varphi_{n+1}$ to take the $\mathcal{E}@\varphi$-match from $(\varphi_n, s)$ to $(\varphi_{n+1}, s)$. In the end, we obtain:
	\begin{proposition} 
	\label{prop:soundness}
	If $\Gamma$ is the conclusion of a $\mathsf{Focus_\infty^2}$-proof, then $\Gamma$ is valid. 
	\end{proposition}
  \subsection{Completeness}
	For completeness we conversely show that from a winning strategy $T$ for Refuter in $\mathcal{G}@\Gamma$, we can construct a model $\mathbb{S}^T$ and a positional strategy $f_T$ for $\forall$ in $\mathcal{E}(\bigwedge \Sigma, \mathbb{S}^T)$ such that $\mathbb{S}^T$ falsifies $\Gamma^-$ with respect to $f_T$. The strategy $f_T$ we construct will not necessarily be optimal, but by Theorem \ref{thm:posdet} of Appendix \ref{sec:apppar} it follows that there must also be an ops $f$ such that $\mathbb{S}^T \not \Vdash_f \Gamma^-$. We will view $T$ as a tree, and restrict attention a certain subtree. We first need to define two relevant properties of rule applications. 
	\begin{definition}
	A rule application is \emph{cumulative} if all of the premisses are supersets of the conclusion. A rule application is \emph{productive} if all of the premisses are distinct from the conclusion. 
	\end{definition}
	Without renaming $T$, we restrict $T$ to its subtree where Prover adheres to the following (non-deterministic) strategy:
	\begin{enumerate}
		\item Exhaustively apply productive instances of $\mathsf{cut}$ and $\mathsf{tc}$.
		\item If applicable, apply the focus rule.
		\item Exhaustively take applications of $\mathsf{R}_\lor$, $\mathsf{R}_\land$, $\mathsf{R}_{\mu}$, $\mathsf{R}_\nu$, $\mathsf{trans}$ that are both cumulative and productive. 
		\item If applicable, apply an axiom.
		\item If applicable, apply a modal rule and loop back to stage (1).
	\end{enumerate}
	It is not hard to see that each of the above phases terminates. More precisely, phases (2), (4) and (5) either terminate immediately or after applying a single rule. By the productivity requirement and the finiteness of $\Sigma$, phases (1) and (3) must terminate after a finite number of rule applications as well. Note also that non-cumulative rule applications can only happen in phases (2) or (5). 

	We will now define the model $\mathbb{S}^T$. The set $S^T$ of states consists of maximal paths in $T$ not containing a modal rule. We write $\Gamma(\rho)$ for $\bigcup\{\Gamma : \Gamma \text{ occurs in } \rho\}$. Note that, since the only possibly non-cumulative rule application in $\rho$ is the focus rule, $\Gamma(\rho)^\f = \mathsf{last}(\rho)^\f$ for every state $\rho$ of $\mathbb{S}^T$. Moreover, we write $\rho_1 \xrightarrow{a} \rho_2$ if $\rho_2$ is directly above $\rho_1$ in $T$, separated only by an application of $\mathsf{R}_{[a]}$ (we assume that trees grow upwards). We write $\rightarrow$ for the union $\bigcup \{\xrightarrow{a} : a \in \mathsf{D}\}$. Clearly, under the relation $\rightarrow$ the states of $\mathbb{S}^T$ form a forest (not necessarily a tree!). We write $\rho \leq \tau$ if $\tau$ is a descendant of $\rho$ in this forest, \emph{i.e.} $\leq$ is the reflexive-transitive closure of $\rightarrow$. The relations $R_a^T$ of $\mathbb{S}^T$ are defined as follows:
	\[
	\rho_1 R_a^T \rho_2 \text{ if and only if } \rho_1 \xrightarrow{a} \rho_2 \text{ or } \rho_2 \xrightarrow{\breve{a}} \rho_1.
	\]
	Note that $\mathbb{S}^T$ is clearly regular. We define the valuation $V^T :S^T \rightarrow \mathcal{P}(\mathsf{P})$ by
	\[
	V^T(\rho) := \{p : \overline p \in \Gamma(\rho)^-\}.
	\]
	The restriction on $T$, together with the fact that it is winning for Refuter, guarantees that each $\Gamma(\rho)$ satisfies certain saturation properties, which are spelled out in the following lemma. We will later use these saturation conditions to construct our positional strategy $f_T$ for $\forall$ in $\mathcal{E}(\bigwedge \Sigma, \mathbb{S}^T)$ and to show that $\mathbb{S}^T$ falsifies $\Gamma$ with respect to $f_T$.
	\begin{lemma}
		\label{lem:saturation}
		For every state $\rho$ of $\mathbb{S}^T$, the set $\Gamma(\rho)$ is \emph{saturated}. That is, it satisfies all of the following conditions:
		\begin{itemize}
			\item For no $\varphi$ it holds that $\varphi, \overline{\varphi} \in \Gamma(\rho)^-$. 
			\item For all $\varphi$ it holds that $\varphi^\u \in \Gamma(\rho)$ if and only if $\overline{\varphi}^\u \notin \Gamma(\rho)$
			\item For all $\varphi$ it holds that $\varphi \leadsto \psi \in \Gamma(\rho)$ if and only if $\varphi \not \leadsto \psi \notin \Gamma(\rho)$. 
			\item For no $\varphi$ it holds that $\varphi \leadsto \varphi \in \Gamma(\rho)$. 
			\item If $\psi_1 \lor \psi_2 \in \Gamma(\rho)^-$, then for both $i$: $\psi_1 \lor \psi_2 \not \leadsto \psi_i \in \Gamma(\rho)$ and $\psi_i \in \Gamma(\rho)^-$. 
			\item If $\psi_1 \land \psi_2 \in \Gamma(\rho)^-$, then for some $i$: $\psi_1 \land \psi_2 \not \leadsto \psi_i \in \Gamma(\rho)$ and $\psi_i \in \Gamma(\rho)^-$. 
			\item If $\mu x \varphi \in \Gamma(\rho)^-$, then $\varphi[\mu x\varphi/x] \in \Gamma(\rho)^-$.
			\item If $\nu x \varphi \in \Gamma(\rho)^-$, then $\nu x \varphi \not \leadsto \varphi[\nu x \varphi/x] \in \Gamma(\rho)$ and $\varphi[\nu x \varphi/x] \in \Gamma(\rho)^-$. 
			\item If $\nu x \varphi \in \Gamma(\rho)^-$, then $\varphi[\nu x \varphi/ x] \leadsto \nu x \varphi \in \Gamma(\rho)$. 
			\item If $\varphi \not \leadsto \psi, \psi \not \leadsto \chi \in \Gamma(\rho)$, then $\varphi \not \leadsto \chi \in \Gamma(\rho)$. 
		\end{itemize}
	\end{lemma}
Now let $\rho_0$ be a state of $\mathbb{S}^T$ containing the root $\Gamma$ and let $\varphi_0$ be some formula such that $\varphi_0 \in \Gamma^-$. We wish to show that $\varphi_0$ is \emph{not} satisfied at $\rho_0$ in $\mathbb{S}^T$. To this end, we will construct a winning strategy $f_T$ for $\forall$ in the game $\mathcal{E} := \mathcal{E}(\bigwedge \Sigma, \mathbb{S}^T)$ initialised at $(\varphi_0, \rho_0)$. The strategy $f_T$ is defined as follows:
	\begin{itemize}
		\item At $(\psi_1 \land \psi_2, \rho)$, pick a conjunct $\psi_i \in \Gamma(\rho)^-$ such that $\psi_1 \land \psi_2 \not \leadsto \psi_i \in \Gamma(\rho)$. 
		\item At $([a]\varphi, \rho)$, choose $(\varphi, \tau)$ for some $\tau$ such that $\rho \xrightarrow{a} \tau$ by virtue of some application of $\mathsf{R}_{[a]}$ with $[a] \varphi^{\a}$ principal for some $b \in \{\u, \f\}$.
	\end{itemize}
	Before we show that $f_T$ is winning for $\forall$, we must first argue that it is well defined. By saturation, for every formula $\psi_1 \land \psi_2$ contained in $\Gamma(\rho)^-$, there is a $\psi_i \in \Gamma(\rho)^-$ with $\psi_1 \land \psi_2 \not \leadsto \psi_i \in \Gamma(\rho)$. Likewise, for every formula $[a] \varphi^{\a} \in \Gamma(\rho)$, there is a $\tau$ directly above $\rho$ in $T$, separated only by an application of $\mathsf{R}_{[a]}$ with $[a] \varphi^{\a}$ principal. The following lemma therefore suffices. Its proof is by induction on the length of $\mathcal{M}$ and heavily relies on the saturation properties of Lemma \ref{lem:saturation}. 
	\begin{lemma}
		\label{lem:inevalimpliesinsequent}
		Let $\mathcal{M}$ be an $f_T$-guided $\mathcal{E}$-match initialised at $(\varphi_0, \rho_0)$. Then for any position $(\varphi, \rho)$ occurring in $\mathcal{M}$ it holds that $\varphi \in \Gamma(\rho)^-$. Moreover, if $(\varphi, \rho)$ comes directly after a modal step and the focus rule is applied in $\rho$, then $\varphi^\f \in \Gamma(\rho)$.
	\end{lemma}
	The following lemma is key to the completeness proof. It shows that if an $f_T$-guided $\mathcal{E}@(\varphi_0, \rho_0)$-match loops from some state $\rho$ to itself, without passing through a $\mu$-formula, then this information is already contained in $\rho$ in the form of a negated trace atom. The proof goes by induction on the number of distinct states of $S^T$ occurring in $\mathcal{N}$. The base case, where only $\rho$ is visited, can be shown by applying several instances of Lemma \ref{lem:saturation}. For the inductive step, we crucially rely on the conditions 2(a) -- 2(d) of Definition \ref{defn:jump} to relate the trace atoms in two states $\tau$ and $\tau'$ such that $\tau R^T_a \tau'$.
	\begin{lemma}
	\label{lem:loops}
	Let $\rho \in S^T$. Suppose that an $f_T$-guided $\mathcal{E}@(\varphi_0, \rho_0)$-match $\mathcal{M}$ has a segment $\mathcal{N}$ of the form:
	\[
	(\varphi, \rho) = (\psi_0, s_0) \cdot (\psi_1, s_1) \cdots (\psi_n, s_n) = (\psi, \rho) \ \ \ (n \geq 0)
	\]
	such that for no $i < n$ the formula $\psi_i$ is a $\mu$-formula.  Then $\varphi \not \leadsto \psi \in \Gamma(\rho)$. 
	\end{lemma}
	With the above lemmata in place, we are ready to prove that $\forall$ wins every full $f_T$-guided $\mathcal{E}@(\varphi_0, \rho_0)$-match $\mathcal{M}$. If $\mathcal{M}$ is finite, it is not hard to show that it must be $\exists$ who got stuck. If $\mathcal{M}$ is infinite, the proof depends on whether $\mathcal{M}$ visits some single state infinitely often. If it does, one can show that if $\exists$ would win the match $\mathcal{M}$, then $\mathcal{M}$ would visit some state $\rho$ with $\nu x \varphi, \varphi[\nu x \varphi / x] \not \leadsto \varphi \in \Gamma(\rho)^-$, contradicting saturation. If, on the other hand, $\mathcal{M}$ visits each state at most finitely often, the proof works by showing that a win for $\exists$ in $\mathcal{M}$ would imply that $T$ contains an infinite branch won by Prover, which is also a contradiction. In the end, we obtain the following proposition. 
	\begin{proposition}
	\label{prop:ftiswinningabelard}
	The strategy $f_T$ is winning for $\forall$ in $\mathcal{E}@(\varphi_0, \rho_0)$.
	\end{proposition}
  Since $\varphi_0$ was chosen arbitrarily from $\Gamma^-$, we find that $\mathbb{S}^T \not \Vdash_{f_T} \Gamma^-$. Hence, by Theorem \ref{thm:posdet} of Appendix \ref{sec:apppar}, we obtain completeness for the formulas in a sequent. 
	\begin{proposition}
	\label{prop:completeness}
	If $\Gamma^-$ is valid, then $\Gamma$ has a $\mathsf{Focus^2_\infty}$-proof. 
	\end{proposition}
	\section{Conclusion}
	We have constructed a non-well-founded proof system $\mathsf{Focus}^\mathsf{2}_\infty$ for the two-way alternation-free modal $\mu$-calculus $\mathcal{L}^{\textit{af}}_{2\mu}$. This system naturally reduces to a cyclic system when restricting to positional strategies in the proof search game. 

	Using the proof search game and the game semantics for the modal $\mu$-calculus, we have shown that the system is sound for all sequents, and complete for those sequents not containing trace atoms. A natural first question for future research is to see if a full completeness result can be obtained. For this, a logic of trace atoms would have to be developed. One could for instance think of a rule like
	\[
	\AxiomC{$\varphi \leadsto \chi, \Gamma$}
	\AxiomC{$\psi \leadsto \chi, \Gamma$}
	\RightLabel{$\mathsf{R}_\land$}
	\BinaryInfC{$\varphi \land \psi \leadsto \chi, \Gamma$}
	\DisplayProof
	\]
	Following on this, we think it would be interesting to properly include trace atoms in the syntax by allowing the Boolean, modal and perhaps even the fixed point operators to apply to trace atoms. An example of a valid formula in this syntax is given by $((\varphi \leadsto \ld{a} \psi) \land \lb{a}(\psi \leadsto \ld{\breve{a}} \varphi)) \rightarrow \varphi$. 

	Another pressing question is whether our system could be used to prove interpolation, as has been done for language without backwards modalities in~\cite{MartiVenema21}. To the best of our knowledge it is currently an open question whether $\mathcal{L}^{\textit{af}}_{2\mu}$ has interpolation. At the same time, it is known that analytic applications of the cut rule do not necessarily interfere with the process of extracting interpolants from proofs~\cite{KowalskiO17,Nguyen01}. 

	Finally, it would be interesting to see if our system can be extended to the full language $\mathcal{L}_{2\mu}$. The main challenge would be to keep track of the most important fixed point variable being unfolded on a trace. Perhaps this could be done by employing an annotation system such as the one by Jungteerapanich and Stirling~\cite{Stirling14,Jungteerapanich09}, together with trace atoms that record the most important fixed point variable unfolded on a loop. 
\subsubsection{Acknowledgements} We thank Johannes Marti for insightful conversations at the outset of the present research. We also thank the anonymous reviewers for their helpful comments. 
\newpage
\appendix
  \section{Parity games}
  \label{sec:apppar}
  \begin{definition}
  A \emph{(two-player) game} is a structure $\mathcal{G} = (B_0, B_1, E, W)$ where $E$ is a binary relation on $B := B_0 + B_1$, and $W$ is a map $B^\omega \rightarrow \{0, 1\}$. 
  \end{definition}
  The set $B$ is called the \emph{board} of $\mathcal{G}$, and its elements are called \emph{positions}. Whether a position belongs to $B_0$ or $B_1$ determines which player \emph{owns} that position. If a player $\Pi \in \{0, 1\}$ owns a position $q$, it is their turn to play and the set of their \emph{admissible moves} is given by the image $E[q]$. 
  \begin{definition}
  A \emph{match} in $\mathcal{G} = (B_0, B_1, E, W)$ (or simply a \emph{$\mathcal{G}$-match}) is a path $\mathcal{M}$ through the graph $(B, E)$. A match is said to be \emph{full} if it is a maximal path.
  \end{definition}
  Note that a full match $\mathcal{M}$ is either finite, in which case $E[\mathsf{last}(\mathcal{M})] = \emptyset$, or infinite. For a $\Pi \in \{0, 1\}$, we write $\overline \Pi$ for the other player $\Pi + 1 \mod 2$.
  \begin{definition}
  A full match $\mathcal{M}$ in $\mathcal{G} = (B_0, B_1, E, W)$ is \emph{won} by player $\Pi$ if either $\mathcal{M}$ is finite and $\mathsf{last}(\mathcal{M}) \in B_{\overline{\Pi}}$, or $\mathcal{M}$ is infinite and $W(\mathcal{M}) = \Pi$. 
  \end{definition}
  If a full match $\mathcal{M}$ is finite, and $\mathsf{last}(\mathcal{M})$ belongs to $B_\Pi$ for $\Pi \in \{0, 1\}$, we say that the player $\Pi$ got \emph{stuck}. A \emph{partial match} is a match which is not full.
  \begin{definition}
  In the context of a game $\mathcal{G}$, we denote by $\textnormal{PM}_\Pi$ the set of partial $\mathcal{G}$-matches $\mathcal{M}$ such that $\mathsf{last}(\mathcal{M})$ belongs to the player $\Pi$. 
  \end{definition}
  \begin{definition}
  A strategy for $\Pi$ in a game $\mathcal{G}$ is a map $f : \textnormal{PM}_\Pi \rightarrow B$. Moreover, a $\mathcal{G}$-match $\mathcal{M}$ is said to be \emph{$f$-guided} if for any $\mathcal{M}_0 \sqsubset \mathcal{M}$ with $\mathcal{M}_0 \in \textnormal{PM}_\Pi$ it holds that $\mathcal{M}_0 \cdot f(\mathcal{M}_0) \sqsubseteq \mathcal{M}$. 
  \end{definition}
  For a position $q$, the set $\textnormal{PM}_\Pi(q)$ contains all $\mathcal{M} \in \textnormal{PM}_\Pi$ such that $\mathsf{first}(\mathcal{M}) = q$. 
  \begin{definition}
  A strategy $f$ for $\Pi$ in $\mathcal{G}$ is \emph{surviving} at a position $q$ if $f(\mathcal{M})$ is admissible for every $\mathcal{M} \in \textnormal{PM}_\Pi(q)$, and \emph{winning} at $q$ if in addition all full $f$-guided matches starting at $q$ are won by $\Pi$. A position $q$ is said to be \emph{winning} for $\Pi$ if $\Pi$ has a strategy winning at $q$. We denote the set of all positions in $\mathcal{G}$ that are winning for $\Pi$ by $\textnormal{Win}_\Pi(\mathcal{G})$.
  \end{definition}
  We write $\mathcal{G}@q$ for the game $\mathcal{G}$ \emph{initialised} at the position $q$ of $\mathcal{G}$. A strategy $f$ for $\Pi$ is \emph{surviving} (\emph{winning}) in $\mathcal{G}@q$ if it is surviving (winning) in $\mathcal{G}$ at $q$.
  \begin{definition}
  A strategy $f$ is \emph{positional} if it only depends on the last move, \emph{i.e.} if $f(\mathcal{M}) = f(\mathcal{M}')$ for all $\mathcal{M}, \mathcal{M}' \in \textnormal{PM}_\Pi$ with $\mathsf{last}(\mathcal{M}) = \mathsf{last}(\mathcal{M}')$. 
  \end{definition}
  We will often present a positional strategy for $\Pi$ as a map $f : B_\Pi \rightarrow B$. 
  \begin{definition}
  A \emph{priority map} on some board $B$ is a map $\Omega : B \rightarrow \omega$ of finite range. A \emph{parity game} is a game of which the winning condition is given by $W_\Omega(\mathcal{M}) = \max (\textit{Inf}_\Omega(\mathcal{M})) \mod 2$, where $\textit{Inf}_\Omega(\mathcal{M})$ is the set of positions occuring infinitely often in $\mathcal{M}$. 
  \end{definition}
  The following theorem captures the key property of parity games: they are \emph{positionally determined}. In fact, each player $\Pi$ has a positional strategy $f_\Pi$ that is \emph{optimal}, in the sense that $f_\Pi$ is winning for $\Pi$ in $\mathcal{G}@q$ for \emph{every} $q \in \text{Win}_\Pi(\mathcal{G})$.
  \begin{theorem}[\cite{mostowski,DBLP:conf/focs/EmersonJ91}]
  \label{thm:posdet}
  For any parity game $\mathcal{G}$, there are positional strategies $f_\Pi$ for each player $\Pi \in \{0, 1\}$, such that for every position $q$ one of the $f_\Pi$ is a winning strategy for $\Pi$ in $\mathcal{G}@q$. 
  \end{theorem}
\bibliographystyle{splncs04}
\bibliography{af_2_focus}
\section{Proofs}
  \noindent \emph{Proof of Proposition \ref{prop:soundness}.} Our proof will go by contraposition, so suppose that some sequent $\Gamma$ is invalid. This means that there is a model $\mathbb{S}$ with a state $s$ and ops $f$ for $\forall$ in the game $\mathcal{E} := \mathcal{E}(\bigwedge \Sigma, \mathbb{S})$, such that $\mathbb{S}, s \not \Vdash_f \Gamma$. We will construct a (positional) winning strategy $T_f$ for Refuter in the game $\mathcal{G} := \mathcal{G}(\Sigma)$ initialised at $\Gamma$.
  
  Formally, this strategy is a function $T_f : \textnormal{PM}_R(\Gamma) \rightarrow \mathsf{Seq}_\Sigma$. In addition, we will define a function $s_f : \textnormal{PM}(\Gamma) \rightarrow \mathbb{S}$, from partial $\mathcal{G}$-matches starting at $\Gamma$ to states of $\mathbb{S}$, such that $\mathbb{S}, s_f(\mathcal{M}) \not \Vdash_f \mathsf{last}(\mathcal{M})$ for every $T_f$-guided $\mathcal{M} \in \text{PM}_P(\Gamma)$, and $\mathbb{S}, s_f(\mathcal{M}) \not \Vdash_f T_f(\mathcal{M})$ for every $T_f$-guided $\mathcal{M} \in \text{PM}_R(\Gamma)$. 
  
  We define $T_f$ and $s_f$ by induction on the length $|\mathcal{M}|$ of a match $\mathcal{M} \in \text{PM}(\Gamma)$. For the base case, \emph{i.e.} where $|\mathcal{M}| = 1$, we have $\mathcal{M} = \Gamma$. Since in this case $\mathcal{M} \in \text{PM}_P(\Gamma)$, we only have to define $s_f(\mathcal{M})$ and not $T_f(\mathcal{M})$. We set $s_f(\mathcal{M}) := s$.
  
  Now suppose that $T_f$ and $s_f$ have been defined for all matches up to length $n$, and that $|\mathcal{M}| = n + 1$. We assume that $\mathcal{M}$ is $T_f$-guided, for otherwise we may just assign $T_f(\mathcal{M})$ and $s_f(\mathcal{M})$ some garbage value.
  
  Suppose first that $\mathcal{M}$ belongs to $\text{PM}_P(\Gamma)$. Writing $\mathcal{M}_{\leq n} \in \text{PM}_R(\Gamma)$ for the initial segment of $\mathcal{M}$ consisting of the first $n$ moves, we set $s_f(\mathcal{M}) := s_f(\mathcal{M}_{\leq n})$. Since $\mathcal{M}$ is $T_f$-guided, we have $\mathsf{last}(\mathcal{M}) = T_f(\mathcal{M}_{\leq n})$. Hence it holds by the induction hypothesis that $\mathbb{S}, s_f(\mathcal{M}) \not \Vdash_f \mathsf{last}(\mathcal{M})$, as required.
  
  If $\mathcal{M}$ belongs to $\text{PM}_R(\Gamma)$, then $\mathsf{last}(\mathcal{M})$ is a rule instance and we distinct cases based on the rule $\mathsf{r}$ of $\mathsf{last}(\mathcal{M}) \in \mathsf{Inst}_\Sigma$.
  \begin{itemize}
    \item $\mathsf{r}$ is an axiom. This can never happen, because then $\mathcal{M}(n)$ would have to be valid while we inductively know that $s_f(\mathcal{M}_{\leq n})$ refutes $\mathcal{M}(n)$.
    \item $\mathsf{r} \in \{\mathsf{R}_\lor, \mathsf{R}_{\mu}, \mathsf{R}_\nu, \mathsf{F}, \mathsf{trans}\}$. In these cases there is only one choice $\Delta \in \mathsf{Seq}_\Sigma$ for Refuter. We set $T_f(\mathcal{M}) := \Delta$ and $s_f(\mathcal{M}) := s_f(\mathcal{M}_{\leq n})$.
    \item $\mathsf{r} = \mathsf{R}_\land$. We set $s_f(\mathcal{M}) := s_f (\mathcal{M}_{\leq n})$ and let $T_f(\mathcal M)$ be the premiss corresponding to $f(\varphi \land \psi, s_f (\mathcal M))$, where $\varphi \land \psi$ is the principal formula of $\mathsf{last}(\mathcal{M})$.
    \item $\mathsf{r} = \mathsf{R}_{[a]}$. In this case we let $s_f(\mathcal{M})$ be the state in $f([a] \varphi, s_f(\mathcal{M}_{\leq n}))$, where $[a] \varphi$ is principal in $\mathsf{last}(\mathcal{M})$. For $T_f$ there is only a single choice, say $\Delta$. We set $T_f(\mathcal{M}) := \Delta$.
    \item $\mathsf{r} \in \{\mathsf{cut}, \mathsf{tc}\}$. First, we set $s_f(\mathcal{M}) := s_f(\mathcal{M}_{\leq n})$. To define $T_f(\mathcal{M})$, note that there are two premisses: $A_1, \Gamma$ and $A_2, \Gamma$. Moreover, by the optimality of $f$, we have $\mathbb{S}, s_f(\Gamma) \Vdash_f A_1$ if and only if $\mathbb{S}, s_f (\Gamma) \not \Vdash_f A_2$. We let $T_f(\mathcal{M})$ be the unique $A_i, \Gamma$ such that $\mathbb{S}$ does \emph{not} satisfy $A_i$ at $s_f (\mathcal{M})$ with respect to $f$. 
  \end{itemize} 
  It is not hard to verify that in each case $\mathbb{S}$ indeed falsifies $T_f(\mathcal{M})$ at $s_f(\mathcal{M})$ with respect to $f$. Also note that $s_f(\mathcal{M})$ almost always equals $s_f (\mathcal{M}_{\leq n})$, with as only possible exception the case where $\mathcal{M}$ belongs to $\text{PM}_R(\Gamma)$ and the rule application of $\mathsf{last}(\mathcal{M})$ is modal.

  We will now show that $T_f$ is indeed a winning strategy for Refuter in $\mathcal{G}@\Gamma$. To that end, suppose towards a contradiction that Refuter loses a $T_f$-guided $\mathcal{G}@\Gamma$-match $\mathcal{M}$. We already know that Refuter does not get stuck, as an axiom is never reached and all other rule instances have a non-zero number of premisses. Hence, the match $\mathcal{M}$ must be infinite, and the rules of $\mathcal{G}$ dictate that there will be a final segment $\mathcal{N} = \Gamma_{0} \cdot i_0 \cdot \Gamma_{1} \cdot i_{1} \cdots$  of $\mathcal{M}$ on which every sequent $\Gamma_n$ has a formula in focus, and the rule instance $i_n$ is modal for infinitely many $n$. We use $\mathcal{K}$ to denote the initial segment of $\mathcal{M}$ occurring before $\mathcal{N}$, \emph{i.e.} such that $\mathcal{M} = \mathcal{K} \cdot \mathcal{N}$. Without loss of generality we assume that $|\mathcal{K}| > 0$. By K\H{o}nig's Lemma, there is a sequence of formulas $\varphi_0, \varphi_{1}, \ldots$ such that for every $n$ it holds that $\varphi_{n}^{\f} \in \Gamma_{n}$ as well as at least one of following:
  \begin{itemize}
    \item $\varphi_{n + 1}^\f \in \Gamma_{n + 1}$ is an immediate ancestor of $\varphi_{n}^\f \in \Gamma_n$;
    \item $i_n = \mathsf{R}_{[a]}$ and $\Gamma_{n}$ contains some $\varphi_{n} \not \leadsto \xi$ such that $\varphi_{n + 1}^\f \in \Gamma_{n+1}$ is an immediate ancestor of some $\xi^{\a} \in \Gamma_n$ with $b \in \{\u, \f\}$. 
  \end{itemize} 
  As before, we write $\mathcal{N}_{\leq n}$ for the initial segment of $\mathcal{N}$ up to the first $n$ moves. Note that $T_f (\mathcal{K} \cdot \mathcal{N}_{\leq 2n}) = \Gamma_n$ for every $n \geq 0$. For convenience we will denote $\mathcal{K} \cdot \mathcal{N}_{\leq 2n}$ by $\mathcal{M}_n$. We will reach a contradiction by showing that $\mathbb{S}, s_f(\mathcal{M}_0) \Vdash_f \varphi_0$, which contradicts the fact that $\mathbb{S}, s_f (\mathcal{M}_0) \not \Vdash_f T_f(\mathcal{M}_0) = \Gamma_0$.

  The crucial claim is that for every $n$ there is an $f$-guided $\mathcal{E}$-match starting at $(\varphi_n, s_f(\mathcal{M}_n))$ and ending at $(\varphi_{n+1}, s_f(\mathcal{M}_{n+1}))$, without passing through a $\mu$-unfolding. More precisely, we will show that there is an $f$-guided $\mathcal{E}$-match
  \[
  (\varphi_n, s_f(\mathcal{M}_n)) = (\psi_0, s_0) \cdots (\psi_m, s_m) = (\varphi_{n+1}, s_f(\mathcal{M}_{n+1})) \ \ \ (m \geq 0)
  \] 
  such that for no $i < m$ the formula $\psi_i$ is a $\mu$-formula. By pasting together these finite segments, it will then follow that the strategy $f$ is not winning for $\forall$ in $\mathcal{E}@(\varphi_0, s_f(\mathcal{M}_0))$, reaching the desired contradiction.

  We will first show the above claim under the assumption that $\varphi_{n+1}^\f$ is an immediate ancestor of $\varphi_n^\f$, and $\varphi_n = \varphi_{n+1}$. In this case $i_n$ is not the modal rule, since the modal rule has no side formulas. Hence $s_f(\mathcal{M}_n) = s_f(\mathcal{M}_{n+1})$ and thus $(\varphi_n, s_f (\mathcal{M}_n)) = (\varphi_{n + 1}, s_f (\mathcal{M}_{n+1}))$, by which the result holds vacuously.

  Now suppose that $\varphi_{n+1}^\f$ is an immediate ancestor of $\varphi_n^\f$ and $\varphi_n \not= \varphi_{n+1}$. We will show, by a case distinction on the main connective of $\varphi_n$, that the match proceeds to the desired position $(\varphi_{n+1}, s_f(\mathcal{M}_{n+1}))$ after a single round.
  \begin{itemize}
    \item First note that $\varphi_n$ cannot be atomic, for atomic formulas can only have immediate ancestors when they are side formulas.
    \item Suppose $\varphi_n$ is of the form $\psi_1 \lor \psi_2$. Then $\varphi_n^\f$ must be principal and we have $\varphi_{n+1} = \psi_i$ for some $i \in \{1, 2\}$. We let $\exists$ simply choose the appropriate disjunct. Since the rule of $i_n$ must be $\mathsf{R}_\lor$, we have $s_f (\mathcal{M}_n)  = s_f (\mathcal{M}_{n+1})$ and thus reach the desired position in $\mathcal{E}$.
    \item Suppose $\varphi_n$ is of the form $\psi_1 \land \psi_2$. Again we find that $\varphi_n^\f$ must be principal, the rule of $i_n$ now being $\mathsf{R}_\land$. By construction we have $\varphi_{n+1} = f(\varphi_n, s_f(\mathcal{M}_n))$, hence the the next position in $\mathcal{E}$ again suffices.
    \item Suppose $\varphi_n = \langle a \rangle \psi$. Then the rule of $i_n$ must be $\mathsf{R}_{[a]}$ and $\varphi_{n + 1} = \psi$. By construction, we have that $s_f (\mathcal{M}_{n +1})$ is the state of $f([a]\chi, s_f(\mathcal{M}_n))$, where $[a]\chi$ is the principal formula of the rule instance $i_n$. Since $s_f(\mathcal{M}_{n+1})$ is an $a$-successor of $s_f(\mathcal{M}_n)$ in $\mathbb{S}$, we can let $\exists$ choose $(\varphi_{n+1}, s_f(\mathcal{M}_{n+1}))$, as required.
    \item If $\varphi_n = [a] \chi$, then the rule of $i_n$ must be $\mathsf{R}_{[a]}$ and $\varphi_n$ must be the principal formula of this rule instance. As $s_f (\mathcal{M}_{n +1})$ is the state of $f([a]\chi, s_f(\mathcal{M}_n))$, the next position in $\mathcal{E}$ will be $(\chi, s_f(\mathcal{M}_{n+1})$, as required.
    \item $\varphi_n = \mu x \psi$ is not possible, because any immediate ancestor of $\mu x \psi^\f$ that is not a side formula, will be out of focus.
    \item Finally, suppose that $\varphi_n = \nu x \psi$. We have that $\varphi_{n+1} = \psi[\nu x \psi/x]$ and the rule of $i_n$ is $\mathsf{R}_\nu$. Because $s_f(\mathcal{M}_{n+1}) = s_f(\mathcal{M}_{n})$, the required position is reached immediately.
  \end{itemize}

  Finally, suppose that $\varphi_{n + 1}^\f$ is not an immediate ancestor of $\varphi^\f$. Then it must be the case that $i_n = \mathsf{R}_{[a]}$ and $\Gamma_{n}$ contains some $\varphi_{n} \not \leadsto \xi$ such that $\varphi_{n + 1}^\f$ is an immediate ancestor of some $\xi^{\a} \in \Gamma_n$. By assumption $\mathbb{S}, s_f(\mathcal{M}_n) \not \Vdash_f \Gamma_n$, and thus in particular $\mathbb{S}, s_f(\mathcal{M}_n) \Vdash_f \varphi_n \leadsto \xi$. Hence $\exists$ can take the $f$-guided match from $(\varphi_n, s_f(\mathcal{M}_n))$ to $(\xi, s_f(\mathcal{M}_n))$ without passing through a $\mu$-unfolding. Since $\xi^{\a}$ has an immediate ancestor (namely $\varphi_{n+1}^\f$), we find that $\xi$ must be of the form $\langle a \rangle \psi$ or of the form $[a] \chi$, where $[a] \chi$ is the principal formula of $i_n$. In either case we can ensure that the next position after $(\xi, s_f (\mathcal{M}_n))$ is $(\varphi_{n+1}, s_f(\mathcal{M}_{n+1}))$ by using the same argument as above for the $\langle a \rangle$ and $[a]$ cases, respectively.
  
  Since the modal rule is applied infinitely often in $\mathcal{M}$, the segments constructed above must infinitely often be nontrivial, \emph{i.e.} of length $> 1$. Hence, we obtain an infinite $f$-guided $\mathcal{E}@(\varphi_0, s_f(\mathcal{M}_0))$-match won by $\exists$, a contradiction. \qed \vspace{0.3cm}

  \noindent \emph{Proof of Lemma \ref{lem:inevalimpliesinsequent}}. Denote the $n$-th position of $\mathcal{M}$ by $(\varphi_n, \rho_n)$. We proceed by induction on $n$. The base case is simply the fact that $\varphi_0 \in \Gamma(\rho_0)^-$. For the induction step, suppose $(\varphi_n, \rho_n)$ is such that $\varphi_n \in \Gamma(\rho_n)^-$, and the next position is $(\varphi_{n+1}, \rho_{n+1})$. We make a case distinction based on the shape of $\varphi_n$. Note that $\varphi_n \not \in \{p, \overline p\}$, for otherwise there would not be a next position $(\varphi_{n + 1}, \rho_{n + 1})$.
  
  If the main connective of $\varphi_n$ is among $\{\lor, \mu, \nu\}$, it follows directly from saturation that $\varphi_{n + 1}$ belongs to $\Gamma(\rho_{n+1})^-$.  If $\varphi_n$ is a conjunction, then $\varphi_{n + 1}$ is the conjunct of $f_T(\varphi_n, \rho)$, which by the definition of $f_T$ belongs to $\Gamma(\rho_{n+1})^-$. 
    
  Now suppose $\varphi_n$ is of the form $\langle a \rangle \psi$. Then $\rho_{n} R_a^T \rho_{n + 1}$, so either $\rho_n \xrightarrow{a} \rho_{n + 1}$ or $\rho_{n + 1} \xrightarrow{\breve{a}} \rho_n$. If $\rho_{n} R_a^T \rho_{n + 1}$ we clearly have $\varphi_{n + 1} = \psi \in \Gamma(\rho_{n+1})^-$, by case 1(b) of Definition \ref{defn:jump}. Moreover, since in particular $\varphi_{n+1}^b \in \mathsf{first}(\rho_{n+1})$, it follows from the restriction on $T$ that in case the focus rule is applied in $\rho_{n+1}$, we have $\varphi_{n+1}^\f \in \Gamma(\rho_{n+1})$. If $\rho_{n + 1} \xrightarrow{\breve{a}} \rho_n$, we argue by contradiction:
  \begin{align*}
    \psi \notin \Gamma(\rho_{n + 1})^- &\Rightarrow \overline{\psi} \in \Gamma(\rho_{n + 1})^- &\text{(Saturation)} \\
    &\Rightarrow [a] \overline{\psi} \in \Gamma (\rho_n)^- &\text{(Case 1(c) of Definition \ref{defn:jump}, $[a] \overline \psi = \overline{\langle a \rangle \psi} \in \Sigma$)} \\
    &\Rightarrow \langle a \rangle \psi \notin \Gamma(\rho_n)^-, &\text{(Saturation)}
  \end{align*}
  which indeed contradicts the inductive hypothesis that $\langle a \rangle \psi \in \Gamma(\rho_n)^-$. Moreover, if the focus rule is applied in $\rho_{n+1}$, we again argue by contradiction. Suppose $\psi^\f \notin \Gamma(\rho_{n+1})$. Then $\rho_{n+1}^-$ does not contain $\psi^\u$ after phase (1), whence we must have $\overline \psi \in \Gamma(\rho_{n+1})^-$. But then saturation gives $\psi \notin \Gamma(\rho_{n+1})^-$, and we can use the same argument as before. Finally, the case where $\varphi_n$ is of the form $[a] \varphi$ is similar to the easy part of the previous case and therefore left to the reader. \qed \vspace{0.3cm}

  \noindent \emph{Proof of Lemma \ref{lem:loops}}. We proceed by induction on the number of distinct states occurring in $\mathcal{N}$. 

  For the base case, we assume that $\rho$ is the only state visited in $\mathcal{N}$. We proceed by induction on the length $n + 1$ of $\mathcal{N}$. For the (inner) base case, where $|\mathcal{N} = 1|$, we have $\mathsf{first}(\mathcal{N}) = (\varphi, \rho) = \mathsf{last}(\mathcal{N})$. By saturation $\varphi \leadsto \varphi \notin \Gamma(\rho)$ and thus $\varphi \not \leadsto \varphi \in \Gamma(\rho)$, as required. For the inductive step, suppose the claim holds for every match up to size $n + 1$. Suppose $|\mathcal{N}| = n + 2$ and consider the final transition $(\chi, \rho) \cdot (\psi, \rho)$ of $\mathcal{N}$. Since the match proceeds after the position $(\chi, \rho)$, but does not move to a new state of $\mathbb{S}^T$, it follows from the irreflexivity of $\mathbb{S}^T$ that the main connective of $\chi$ must be among $\{\lor, \land, \mu, \nu\}$. Moreover, by Lemma \ref{lem:inevalimpliesinsequent}, we have $\chi \in \Gamma(\rho)^-$. We claim that $\chi \not \leadsto \psi \in \Gamma(\rho)^-$. When the main connective of $\chi$ is in $\{\lor, \mu, \nu\}$, this follows directly from saturation. If $\chi$ is a conjunction, we have, since $\mathcal{M}$ is $f_T$-guided, that $(\psi, \rho) = f_T(\chi, \rho)$. By the definition of $f_T$, it follows that $\chi \not \leadsto \psi \in \Gamma(\rho)$, as required. We finish the proof of this special case of the lemma by applying the induction hypothesis to the initial segment of $\mathcal{N}$ obtained by removing the last position $(\psi, \rho)$. This gives $\varphi \not \leadsto \chi \in \Gamma(\rho)$, hence by saturation $\varphi \not \leadsto \psi \in \Gamma(\rho)$. 
  
  For the (outer) inductive step, suppose that $n > 1$ states are visited in $\mathcal{N}$. We write $\mathcal{N}$ as $\mathcal{A}_1 \cdot \mathcal{B}_1 \cdot \mathcal{A}_2 \cdot \mathcal{B}_2 \cdots \mathcal{A}_m$,
  where for every $(\chi, \tau)$ in $\mathcal{A}_i$ it holds that $\tau = \rho$ and for every $(\chi, \tau)$ in $\mathcal{B}_i$ it holds that $\tau \not= \rho$. As $\mathbb{S}^T$ is a forest, there must for each $\mathcal{B}_i$ be some $\gamma_i$, $\delta_i$, and $\tau_i$ such that $\mathsf{first}(\mathcal{B}_i) = (\gamma_i, \tau_i)$ and $\mathsf{last}(\mathcal{B}_i) = (\delta_i, \tau_i)$. Denote $\mathsf{first}(\mathcal{A}_i) = (\alpha_i, \rho)$ and $\mathsf{last}(\mathcal{A}_i) = (\beta_i, \rho)$. Summing up, we will we use the following notation for each $i \in [1, m)$: 
   \begin{align*}
    \mathsf{first}(\mathcal{A}_i) = (\alpha_i, \rho), &&
    \mathsf{last}(\mathcal{A}_i) = (\beta_i, \rho), &&
    \mathsf{first}(\mathcal{B}_i) = (\gamma_i, \tau_i), &&
    \mathsf{last}(\mathcal{B}_i) = (\delta_i, \tau_i).
  \end{align*} 
  Let $i \in [1, m)$ be arbitrary. Since $\mathcal{B}_i$ does not visit $\rho$, it must visit strictly less states than $\mathcal{N}$. By the induction hypothesis we find that $\gamma_i \not \leadsto \delta_i \in \Gamma(\tau_i)$. We claim that $\alpha_i \not \leadsto \beta_{i+1} \in \Gamma(\rho)$. Since the match $\mathcal{N}$ transitions from the state $\rho$ to the state $\tau_i$, there must be some $a \in \mathsf{D}$ such that $\rho R^T_a \tau_i$. 

We first assume that $\rho \xrightarrow{a} \tau_i$. Then by the nature of the game, $\beta_i$ must be of the form $\beta_i = \langle a \rangle \gamma_i$ or of the form $\beta_i = [a] \gamma_i$, and, since by definition $f_T$ only moves upwards in $\mathbb{S}^T$, we must have $\delta_i = \langle \breve a \rangle \alpha_{i+1}$. We only cover the case where $\beta_i = [a] \gamma_i$ (the case where $\beta_i = \langle a \rangle \gamma_i$ is almost the same, but uses 2(c) instead of 2(a) of Definition \ref{defn:jump}). We indeed find:
  \begin{align*}
      & \gamma_i \not \leadsto \langle \breve{a} \rangle \alpha_{i+1} \in \Gamma(\tau_i) &&\text{(Induction hypothesis, $\delta_i = \ld{\breve{a}} \alpha_{i+1}$)}\\
      \Rightarrow \ &\gamma_i \leadsto \langle \breve{a} \rangle \alpha_{i+1} \notin \Gamma(\tau_i) &&\text{(Saturation)} \\
      \Rightarrow \ &[a]\gamma_i \leadsto \alpha_{i+1} \notin \Gamma(\rho) &&\text{(Case 2(a) of Definition \ref{defn:jump})} \\
      \Rightarrow \ &\beta_i \not \leadsto \alpha_{i+1} \in \Gamma(\rho), &&\text{(Saturation, $\beta_i = \lb{a} \gamma_i$)}
    \end{align*}
Now suppose that $\tau_i \xrightarrow{\breve{a}} \rho$. Then $\beta_i$ must be of the form $\beta_i = \langle a \rangle \gamma_i$, because the strategy $f_T$ moves only upwards in $\mathbb{S}^T$. Moreover, we have $\delta_i = [\breve a] \alpha_{i+1}$ or $\delta_i = \langle \breve a \rangle \alpha_{i + 1}$. An argument similar to the one above, respectively using cases 2(b) and 2(d) of Definition \ref{defn:jump}, shows that $\ld{a} \gamma_i \not \leadsto \alpha_{i+1} \in \Gamma(\rho)$. 

  Applying the induction hypothesis to the $\mathcal{A}_i$, we have $\alpha_i \not \leadsto \beta_i \in \Gamma(\rho)$ for every $1 \leq i \leq m$. Hence, by saturation, we find $\gamma_1 \not \leadsto \delta_{m} \in \Gamma(\rho)$, as required. \qed \vspace{0.3cm}

  \noindent \emph{Proof of Proposition \ref{prop:ftiswinningabelard}}. Let $\mathcal{M}$ be an arbitrary $f_T$-guided and full $\mathcal{E}$-match. By positional determinacy, we may without loss of generality assume that $\exists$ adheres to some positional strategy in $\mathcal{M}$. First suppose that $\mathcal{M}$ is finite. We consider the potential cases one-by-one.

  If $\varphi$ is a propositional letter $p$, we find:
  \[
  \varphi = p \Rightarrow p \in \Gamma(\rho)^- \Rightarrow \overline p \notin \Gamma(\rho)^- \Rightarrow \mathbb{S}^T, \rho \not \Vdash p,
  \]
  where the first implication holds due to Lemma \ref{lem:inevalimpliesinsequent}, the second due to saturation, and the third by the definition of the valuation function of $\mathbb{S}^T$. It follows that in this case $\exists$ gets stuck. 

  Similarly, if $\varphi$ is a negated propositional letter $\overline p$, we find:
  \[
  \varphi = \overline p \Rightarrow \overline p \in \Gamma(\rho)^- \Rightarrow \mathbb{S}^T, \rho \Vdash p \Rightarrow \mathbb{S}^T, \rho \not \Vdash \overline p,
  \]
  hence again $\exists$ gets stuck. 

  Finally, we claim that $\varphi$ is not of the form $[a]\psi$. Indeed, in that case the fact that $[a] \psi \in \Gamma(\rho)^-$ would entail that the modal rule is applicable. Hence $f_T(\varphi, \rho)$ would be defined, contradicting the assumed fullness of $\mathcal{M}$.

  Now suppose that $\mathcal{M}$ is infinite, say $\mathcal{M} = (\varphi_n , \rho_n)_{n \in \omega}$. Suppose first that some state $\rho$ is visited infinitely often in $\mathcal{M}$. By the pigeonhole principle, there must be a formula $\varphi$ and segment $\mathcal{N}$ of $\mathcal{M}$ such that $\mathsf{first}(\mathcal{N}) = \mathsf{last}(\mathcal{N}) = (\varphi, \rho)$. Since both players follow a positional strategy, we can write the match $\mathcal{M}$ as $\mathcal{K} \mathcal{N}^*$, where $\mathcal{K}$ is some initial segment of $\mathcal{M}$. But this means that only finitely many states of $\mathbb{S}^T$ occur in $\mathcal{M}$. As $\mathcal{M}$ is winning for $\exists$, there must, by Proposition \ref{prop:propertiesaf}, be some formula $\nu x \psi$ occurring infinitely often in $\mathcal{M}$. Therefore, there must be a position $(\nu x \psi, \tau)$ occurring infinitely often in $\mathcal{M}$. But then Lemma \ref{lem:loops} gives $\varphi[\nu x \varphi/ x] \not \leadsto \nu x \psi \in \Gamma(\tau)$, contradicting saturation.

  Hence we may assume that $\mathcal{M}$ visits each state $\rho$ at most finitely often. Suppose, towards a contradiction, that $\mathcal{M}$ is won by $\exists$. Let $(\varphi_{\alpha(0)}, \rho_{\alpha(0)})$ be a position of $\mathcal{M}$ after which every unfolding is a $\nu$-unfolding, and $\rho_n > \rho_{\alpha(0)}$ for every $n > \alpha(0)$. Recursively let $\alpha({i + 1})$ be the least index greater than $\alpha(i)$ such that for every $m > \alpha({i + 1})$ it holds that $\rho_m > \rho_{\alpha(i+1)}$. 

  It is not hard to see that for each $i$ there is an $a_i \in \mathsf{D}$ with $\rho_{\alpha(i)} \xrightarrow{a_i} \rho_{\alpha(i+1)}$. This gives a $T$-guided $\mathcal{G}$-match $\mathcal{K} = \rho_{\alpha(0)} \cdot \mathsf{R}_{[a_{\alpha(0)}]} \cdot \rho_{\alpha(1)} \cdot \mathsf{R}_{[a_{\alpha(1)}]} \cdot \rho_{\alpha(2)} \cdot \mathsf{R}_{[a_{\alpha(2)}]} \cdots$. Note that $\mathcal{K}$ is infinite, as $\mathcal{M}$ visits infinitely many states. Because $T$ is by assumption winning for Refuter, the focus rule must be applied infinitely often. 

  Let $\rho_{\alpha(i)}$ with $i > 0$ be a segment on which the focus rule is applied. Note that $\varphi_{\alpha(i) - 1}$ is modal, hence we obtain by Lemma \ref{lem:inevalimpliesinsequent} that $\varphi_{\alpha(i)}^\f \in \Gamma(\rho_{\alpha(i)})$. We claim that for every $j > i$ it holds that every sequent in $\rho_{\alpha(j)}$ has a formula in focus. With this we reach the desired contradiction, because it means that the focus rule cannot be applied on this final segment of $\mathcal{K}$ after all. 

  In particular, we will show that $\varphi_{\alpha(j)}^\f \in \mathsf{first}(\rho_{\alpha(j)})$ for every $j > i$, which suffices by the restriction of $T$ to cumulative rule applications. We proceed by induction on $j - i$. For the base case, we wish to show that $\varphi_{\alpha(i+1)}^\f \in \mathsf{first}(\rho_{\alpha(i + 1)})$. To that end, consider $\mathcal{J} = (\varphi_{\alpha(i)}, \rho_{\alpha(i)}) \cdots (\varphi_{\alpha(i+1) - 1}, \rho_{\alpha({i + 1}) -1})$. Since $T$ is a forest, we have $\rho_{\alpha(i+1) - 1} = \rho_{\alpha(i)}$ and thus either $\alpha(i) = \alpha({i + 1}) - 1$, in which case, by saturation $\varphi_{\alpha(i)} \not \leadsto \varphi_{\alpha(i + 1) - 1} \in \Gamma(\rho_{\alpha(i + 1) - 1})$, or $|\mathcal{J}| > 1$ and we may apply Lemma \ref{lem:loops} to again obtain $\varphi_{\alpha(i)} \not \leadsto \varphi_{\alpha(i+1) - 1} \in \Gamma(\rho_{\alpha(i)})$. Since $\rho_{\alpha(i)} \xrightarrow{a_i} \rho_{\alpha(i+1)}$, it follows that $\varphi_{\alpha(i+1) - 1}$ must be of the form $\langle a_{i} \rangle \varphi_{\alpha(i+1)}$ or of the form $[a_{i}] \varphi_{\alpha(i+1)}$. In either case, Definition \ref{defn:jump}.1 gives $\varphi_{\alpha(i+1)}^\f \in \mathsf{first}(\rho_{\alpha(i+1)})$, as required. 

  For the induction step we can use precisely the same argument. \qed
\end{document}